\documentclass{rspublic}
\usepackage{amsfonts,amssymb}
\usepackage{epsfig}
\usepackage{url}

\begin{document}

\title[ Zeros of Angular Lattice Sums]{The Riemann Hypothesis for Angular
Lattice Sums}
\author[R.C. McPhedran and others]{Ross C. McPhedran$^1$, \\
Lindsay C. Botten$^2$, Dominic J. Williamson$^{1}$ and Nicolae-Alexandru P. Nicorovici$^{1}$}

\affiliation{$^1$CUDOS, School of Physics, University of Sydney, NSW 2006, Australia,\\
 $^2$School of Mathematical Sciences, University of Technology,
Sydney, N.S.W. 2007 Australia}

\label{firstpage}

\maketitle

\begin{abstract}{Lattice sums, Dirichlet $L$ functions, Riemann hypothesis}
We present further results on a class of sums which involve complex
powers of the distance to points in a two-dimensional square lattice
and trigonometric functions of their angle, supplementing those in a
previous paper (McPhedran {\em et al}, 2008). We give a general
expression which permits numerical evaluation of members of the
class of sums to arbitrary order. We use this to illustrate
numerically the properties of trajectories along which the real and
imaginary parts of the sums are zero, and we show results for the
first two of a particular set of angular sums (denoted ${\cal C}(1,4 m;s)$)  which indicate their
density of zeros on the critical line of the complex exponent is the
same as that for the product (denoted ${\cal C}(0,1;s)$) of the Riemann zeta function and the
Catalan beta function. We then introduce a function which is the quotient of the angular lattice sums
${\cal C}(1,4 m;s)$ with ${\cal C}(0,1;s)$, and use its properties to prove that  ${\cal C}(1,4 m;s)$ obeys the
Riemann hypothesis for any $m$ if and only if  ${\cal C}(0,1;s)$ obeys the Riemann hypothesis. We furthermore prove
that if the Riemann hypothesis holds, then ${\cal C}(1,4 m;s)$ and ${\cal C}(0,1;s)$ have the same distribution of zeros
on the critical line (in a sense made precise in the proof).
\end{abstract}

\section{Introduction}
This paper adds to results  in McPhedran {\em et al} (2008)
(hereafter referred to as I) on the properties of a class of sums
over two-dimensional lattices involving trigonometric functions of
the angle to points in the lattice, and a complex power $2 s$ of
their distance from the lattice origin. There, it was shown that certain of these angular
sums had zeros on the critical line $\Real (s)=1/2$, but could not
have zeros in a neighbourhood of it.

We derive a general expression which is exponentially convergent and
permits the rapid and accurate evaluation of the angular sums
irrespective of the value of the complex parameter $s$. We demonstrate the
high-order convergence of this formula by using it to illustrate a
limiting formula for a particular set of angular lattice sums. We go
on to consider the properties of trajectories along which the real
and imaginary parts of a class of angular sums are zero, and in
particular we establish accurate approximations for these
trajectories when $\Real (s)$ lies well outside the critical strip
$0<\Real(s)<1$. We give preliminary results on the distribution of
zeros on the critical line $\Real (s)=1/2$ of two angular sums, which
suggest that to leading order they have the same density of zeros as
the product of the Riemann zeta function and the Catalan beta
function.  In the last section of the main body of the paper, we introduce a new quotient function,
and from its properties we prove that the Riemann hypothesis for the product of the Riemann zeta function
and the Catalan beta function is equivalent to that for all members of a
particular infinite set of angular lattice sums. We also show that when this Riemann hypothesis holds, all members of the class of angular lattice sums
have the same distribution function for zeros on the critical line, and that distribution function also applies to the product of the Riemann zeta function and the Catalan beta
function. In Appendices we comment on the functional equation
satisfied by a class of angular sums, and the properties which link
angular sums of order up to ten.

The analytic results presented here are supported by numerical
results obtained using Mathematica 7.0.1. Formal proofs of certain
key properties of the angular sums and the location of their zeros
will be given here, while others will be presented in a companion
paper.

There are two principal motivations for the study presented here.
The first is that the very general expression for the class of
angular lattice sums derived in Section 2, and their connections
with other angular lattice sums shown in Appendix A, enables them to
be used in physical applications requiring regularization of sums
over a two dimensional square lattice. The class  of summands which
can be addressed is wide, as it consists of any function which has a
Taylor series in integer powers of trigonometric functions of the
lattice point angle in the plane, and any complex power of its
distance from the origin.  For example, any summand of the type
often encountered in solid state physics combining a Bloch-type
phase factor and a function of distance having a Taylor series could
be so represented. The second is that we show the angular lattice
sums to be connected with the product of the Riemann zeta function
and the Catalan beta function in a very natural way- for example, it
seems that the densities of their zeros on the critical line are the
same to leading orders. In this way, these angular sums may provide
a new way to forge a link between the Riemann hypothesis and its
generalization to other Dirichlet $L$ functions, as well as
providing a wide class of functions, in which identified members
obey {\em a priori} the hypothesis, and others do not. There are
interesting parallels between this work and that of S. Gonek (2007),
although we deal with double sums and Gonek with single sums.

It should be noted that the proofs presented are not of great technical difficulty, and yet some of the
results are striking. It thus seems that the study of angular lattice sums may offer a favourable context
in which properties related to the Riemann hypothesis may be studied. It is the authors' hope that the results given
will thus stimulate more mathematicians to recommence the investigation of doubly periodic (and even multiply periodic)
lattice sums-  an area which seems to have been unduly neglected for a considerable period.
\section{An absolutely-convergent expression for angular lattice
sums}
We recall the definition from (I) of two sets of angular
lattice sums for the square array:
\begin{equation}
{\cal C}(n,m;s)=\sum_{p_1,p_2}{'}\frac{\cos^n(m
\theta_{p_1,p_2})}{(p_1^2+p_2^2)^s},~~ {\cal
S}(n,m;s)=\sum_{p_1,p_2}{'}\frac{\sin^n(m
\theta_{p_1,p_2})}{(p_1^2+p_2^2)^s}, \label{mz-1}
\end{equation}
where $\theta_{p_1,p_2}=\arg (p_1+\ri p_2)$, and the prime denotes the
exclusion of the point at the origin. The sum independent of the
angle $\theta_{p_1,p_2}$ was evaluated by Lorenz (1871) and Hardy(1920)
in terms of the product of Dirichlet $L$ functions:
\begin{equation}
{\cal C}(0,m;s)={\cal S}(0,m;s)\equiv{\cal C}(0,1;s)=4 L_1(s)
L_{-4}(s)=4 \zeta(s) L_{-4}(s) .\label{mz2}
\end{equation}
Here $L_1(s)$ is more commonly referred to as the Riemann zeta
function, and $L_{-4}(s)$ as the Catalan beta function.  A useful
account of the properties of Dirichlet $L$ functions has been given
by Zucker \& Robertson (1976).

It is convenient to use a subset of the angular sums (\ref{mz-1}) as
a basis for numerical evaluations. We note that the sums ${\cal
C}(n,1;s)$ are zero if $n$ is odd. We next derive the following
relationship for the non-zero sums ${\cal C}(2 n,1;s)$:
\begin{eqnarray}
\sum_{(p_1,p_2)} {'}\frac{p_1^{2 n}}{(p_1^2+p_2^2)^{s+n}}  = {\cal
C}(2 n,1;s) =\frac{2 \sqrt{\pi} \Gamma(s+n-1/2)\zeta(2
s-1)}{\Gamma(s+n)} \nonumber \\
 +  \frac{8\pi^s}{\Gamma(s+n)} \sum_{p_1=1}^\infty
\sum_{p_2=1}^\infty  \left(
\frac{p_2}{p_1}\right)^{s-1/2}
p_1^n p_2^n \pi^n K_{s+n-1/2}(2\pi p_1 p_2),
\label{gr1}
\end{eqnarray}
where $K_\nu(z)$ denotes the modified Bessel function of the second
kind, or Macdonald function, with order $\nu$ and argument $z$. The
general form  (\ref{gr1}) may be derived following Kober (1936) in
the usual way: a Mellin transform is used to give
\begin{equation}
\sum_{(p_1,p_2)} {'}\frac{p_1^{2
n}}{(p_1^2+p_2^2)^{s+n}}=\sum_{(p_1,p_2)}{'}\frac{p_1^{2
n}}{\Gamma(s+n)}\int_0^\infty t^{s+n-1} \re^{-t(p_1^2+p_2^2)} \rd t.
\label{gr2}
\end{equation}
The Poisson summation formula is then used to transform the sum over
$p_2$, giving
\begin{equation}
\sum_{(p_1,p_2)} {'}\frac{p_1^{2
n}}{(p_1^2+p_2^2)^{s+n}}=\sum_{(p_1,p_2)} {'}\frac{p_1^{2
n}}{\Gamma(s+n)}\int_0^\infty t^{s+n-1} \re^{-t p_1^2}
\sqrt{\frac{\pi}{t}} \re^{-\pi^2 p_2^2/t} \rd t. \label{gr3}
\end{equation}
We then separate the axial contribution, which for $n\neq 0$ comes
from $p_2=0$ alone, and use Hobson's integral
\begin{equation}
\int_0^\infty t^{s-1} \re^{-p t -q/t} \rd t=2\left(
\frac{q}{p}\right)^{s/2} K_s(2\sqrt{q p}) \label{gr4}
\end{equation}
on the remaining double sum. This leads directly to (\ref{gr1}).

It should be noted that the double sum in (\ref{gr1}) is
exponentially convergent. Indeed, from relation 9.7.2 in Abramowitz
and Stegun (1972), the large argument approximation for the
Macdonald function of order $\nu$ is
\begin{equation}
K_\nu(z)\sim \sqrt{\frac{\pi}{2 z}}e^{-z} .\label{macfn1}
\end{equation}
This means that the double sum starts to converge rapidly as soon as
the argument $2\pi p_1 p_2$ everywhere exceeds the modulus of the
order $s+n-1/2$. In practice, accurate answers are achieved when
sums are carried out over $p_1$ and $p_2$ from 1 to $P$, where
$P\sim |s+n-1/2|/\pi$ (the precise value of $P$ required being fixed
by studies of the effect of increasing $P$ on the stability of the
result). The representation (\ref{gr1}) and finite combinations of
it thus furnish absolutely convergent representations of
trigonometric sums from the family (\ref{mz-1}) and close relatives,
for any values of $s$ with finite modulus. These representations are
easily represented numerically in any computational system
incorporating routines for the Riemann zeta function of complex
argument, and Macdonald functions of complex order and real
argument. (In practice, the Macdonald function evaluations are most
time-expensive in the region of $(p_1,p_2)$ values where $ p_1 p_2$
has comparable magnitude to $|s+n-1/2|/(2 \pi)$. Thus, it is
efficient to create a table of these values for $p_1$ varying with
$p_2=1$, running up to an argument where the Macdonald function is
less than an appropriate tolerance times its value for $p_1=1$.)

As an example of the numerical efficacy of (\ref{mz-1}), we consider
its use in illustrating a limiting property of the sums ${\cal C}(2
m,1;s)$.We have
\begin{equation}
{\cal C}(2 m,1;s)=\sum_{(p_1,p_2)} {'}\frac{\cos^{2
m}\theta_{p_1,p_2}}{(p_1^2+p_2^2)^s}, \label{sn26}
\end{equation}
and as $m\rightarrow \infty$ we require $|\cos \theta_{p_1,p_2}|=1$
for a contribution, i.e.
\begin{equation}
\lim_{m\rightarrow \infty}{\cal C}(2 m,1;s)=2\zeta (2 s).
\label{sn27}
\end{equation}
The relationship (\ref{sn27}) is illustrated numerically in Fig.
\ref{supp3}. For the right-hand side of (\ref{sn27}) to be accurate,
the required order $m$ increases with $t$, although convergence is
also slow for $t$ near zero.
\begin{figure}
\begin{center}
\includegraphics[width=6cm]{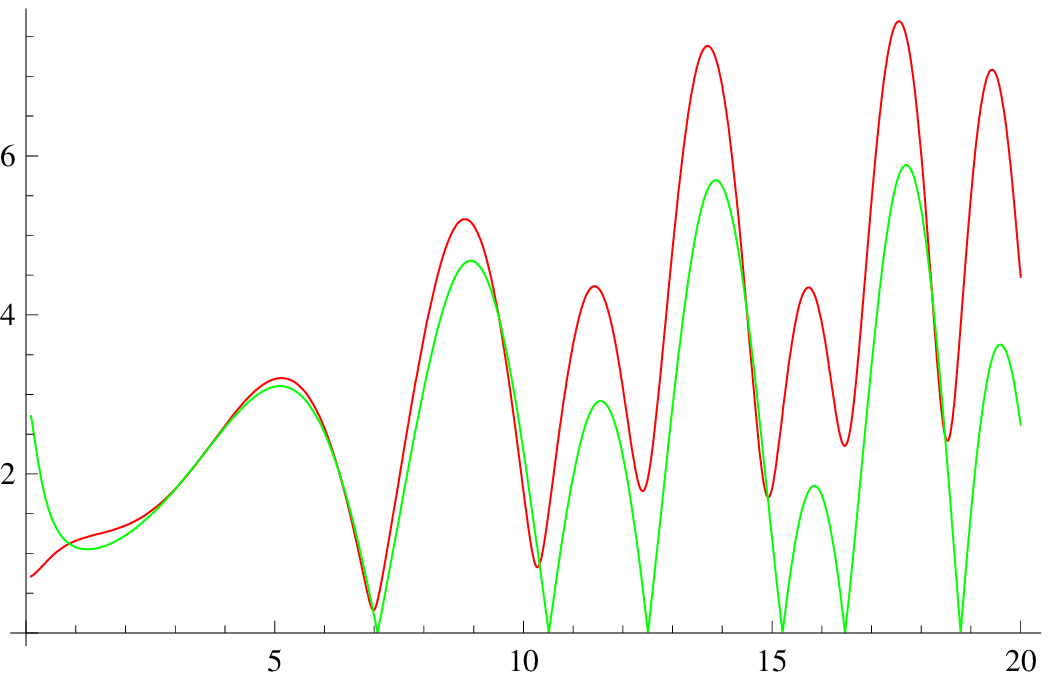}~~
\includegraphics[width=6cm]{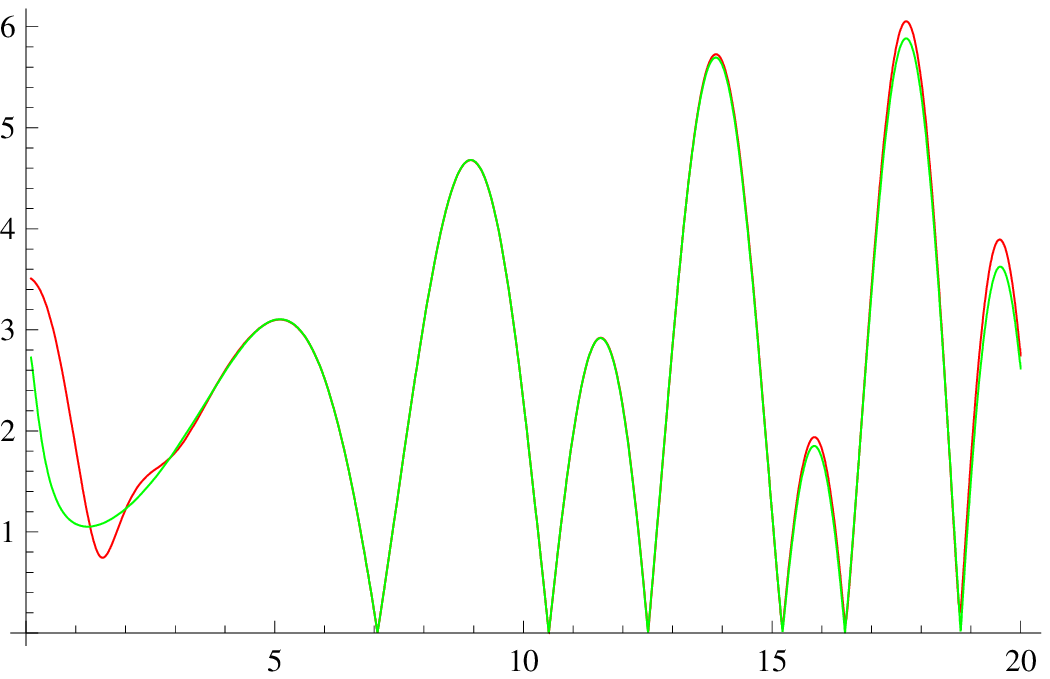}
\end{center}
\caption{The modulus of  ${\cal C}(2 m,1;s)$ (red) and $2 \zeta (2
s)$ (green)  as a function of $s=1/2+\ri t$, for $t \in [0, 20]$,
with $m=10$ (left) and $m=100$ (right).} \label{supp3}
\end{figure}

Another angular sum of great importance in this paper can easily be
expanded in terms of the ${\cal C}(2 n,1;s)$:
\begin{equation}
{\cal C}(1,4 m;s)=\sum_{p_1,p_2} {'}\frac{\cos (4 m
\theta_{p_1,p_2})}{(p_1^2+p_2^2)^s}={\cal C}(2,2 m;s)-{\cal S}(2,2
m;s), \label{mz3}
\end{equation}
or, in terms of the Chebyshev polynomial of the first kind
(Abramowitz \& Stegun (1972), Chapter 22),
\begin{equation}
{\cal C}(1,4 m;s)=\sum_{p_1,p_2} {'}\frac{T_{4 m}(\cos
\theta_{p_1,p_2})}{(p_1^2+p_2^2)^s}. \label{mz3a}
\end{equation}
As the coefficients of this Chebyshev polynomial are explicitly
known, the representation (\ref{mz3a}) enables any sum ${\cal C}(1,4
m;s)$ to be expressed as a linear combination of sums ${\cal C}(2
n,1;s)$ with $0\leq n \leq 2 m$.

The angular sums just introduced may also be expressed as sums along rays in the square lattice.
This means that the particular independent values of $\theta_{p_1,p_2}$ occurring in the lattice are identified, and sums along those directions are performed, giving a multiplying factor
of $\zeta (2 s)$.  The sums along rays may be reduced to those occurring in the first octant, which
we define as follows:
\begin{equation}
{\cal O}_1=\{p_1,p_2|p_1=1,\ldots, \infty, 1\leq p_2 <p_1, p_2\perp p_1\} ,
\label{ray1}
\end{equation}
where the notation $p_2\perp p_1$ means the integers are relatively prime. The axis $p_2=0$
and the line $p_1=p_2$ are dealt with separately, giving
\begin{equation}
{\cal C}(0,1;s)=4 \zeta (2 s) \left[  1 +\frac{1}{2^s} +2\sum_{{\cal O}_1} \frac{\cos^{2 s} \theta_{p_1,p_2}}
{p_1^{2 s}}\right].
\label{ray2}
\end{equation}
The sum in (\ref{ray2}) converges down to $\Re(s)=1$, where it diverges, as is to be expected from
(\ref{mz2}). Other expressions of this type are:
\begin{equation}
{\cal C}(1,4 m;s)=4 \zeta (2 s) \left[  1 +\frac{\cos (m\pi)}{2^s} +2\sum_{{\cal O}_1} \frac{
\cos (4 m\theta_{p_1,p_2}) \cos^{2 s} \theta_{p_1,p_2}}
{p_1^{2 s}}\right],
\label{ray3}
\end{equation}
\begin{equation}
{\cal C}(2,2 m;s)=4 \zeta (2 s) \left[  1 +\frac{\cos^2 (m\pi/2)}{2^s} +2\sum_{{\cal O}_1} \frac{
\cos^2 (2 m\theta_{p_1,p_2}) \cos^{2 s} \theta_{p_1,p_2}}
{p_1^{2 s}}\right],
\label{ray4}
\end{equation}
and 
\begin{equation}
{\cal S}(2,2 m;s)=4 \zeta (2 s) \left[  \frac{\sin^2 (m\pi/2)}{2^s} +2\sum_{{\cal O}_1} \frac{
\sin^2 (2 m\theta_{p_1,p_2}) \cos^{2 s} \theta_{p_1,p_2}}
{p_1^{2 s}}\right].
\label{ray4a}
\end{equation}
The last of these formula we give may be compared with (\ref{sn27}):
\begin{equation}
{\cal C}(2m,1;s)={\cal S}(2 m,1;s)=2 \zeta (2 s) \left[  1+\frac{1}{2^{s+m-1}} +2\sum_{{\cal O}_1} \frac{
(\cos^{2 m} (\theta_{p_1,p_2})+\sin^{2 m} (\theta_{p_1,p_2})) \cos^{2 s} \theta_{p_1,p_2}}
{p_1^{2 s}}\right].
\label{ray4b}
\end{equation}

There is a sum rule connecting the $C(1, 4 m;s)$:
\begin{equation}
\sum_{m=1}^\infty {\cal C}(1,4 m;s)=\pi \zeta (2 s)-\frac{1}{2} {\cal C}(0,1;s).
\label{sumrule}
\end{equation}
This follows from the basic form of the Poisson summation formula (relating a sum of cosine functions to a sum of delta functions).

The connections between various angular lattice sums  grouped in
systems with order up to 10 are explored in Appendix A.

\section{Some properties of trigonometric lattice sums}
The functional equation is known (see McPhedran {\em et al.} (2004),
eqs. 32 and 59) for ${\cal C}(1,4 m;s)$:
\begin{equation}
G_{4 m}(s)={\cal C}(1,4 m;s)\frac{\Gamma(s+2 m)}{\pi^s}=G_{4
m}(1-s). \label{mz4}
\end{equation}
This equation also holds for $m=0$, where it gives the functional
equation for the product $\zeta(s) L_{-4}(s)$. It is in fact the $m$
dependence of the functional equation (\ref{mz4}) which enables the
derivation of many of the results in (I) and the present paper. As
the derivation in McPhedran {\em et al} (2004) uses different
notation to that in subsequent papers and here, we give a brief
discussion of the argument leading to (\ref{mz4}) in Appendix B.

This $m$ dependence in (\ref{mz4}) is represented in two related
functions:
\begin{equation}
{\cal F}_{2 m}(s)=\frac{\Gamma(1-s+2 m) \Gamma(s)}{\Gamma(1-s)
\Gamma(s+ 2m)}=\exp(2 \ri \phi_{2 m}(s)), \label{mz5}
\end{equation}
where $\phi_{2 m}(s)$ is in general complex. Note that ${\cal F}_{2
m}(s)$ is the ratio of two polynomials of degree $2 m$, with one
obtained from the other by replacing $s$ by $1-s$:
\begin{equation}
{\cal F}_{2 m}(s)=\frac{(2 m-s) (2 m-1-s)\ldots (1-s)}{[2 m-(1-s)] [2m-1-(1-s)]\ldots [1-(1-s)]}.
\label{mz5a}
\end{equation}

We then introduce two sets of rescaled lattice sums:
\begin{equation}
\widetilde{C}(2,2 m;s)=\frac{\Gamma(s)}{2\pi^s \sqrt{{\cal F}_{2
m}(s)}}[C(0,1;s)+C(1,4 m;s)], \label{mz6}
\end{equation}
and
\begin{equation}
\widetilde{S}(2,2 m;s)=\frac{\Gamma(s)}{2\pi^s \sqrt{{\cal F}_{2
m}(s)}}[C(0,1;s)-C(1,4 m;s)]. \label{mz6a}
\end{equation}
Note that this definition means $\widetilde{C}$ and $\widetilde{S}$ have
branch cuts where ${\cal F}_{2 m}(s)$ is real and negative. For
example, for ${\cal F}_{2}(s)$, the branch cut includes a circle in
the $(\sigma,t)$ plane ($s = \sigma + \ri t$), with centre $(1/2,0)$ and radius
$\sqrt{3}/2$.

 An analytic function which combines lattice sums is defined by
 \begin{equation}
\Delta_3(2, 2m;s)=\widetilde{\cal C}(2,2 m;s)^2 -\widetilde{\cal S}(2,2
m;s)^2=\frac{\Gamma(s)^2}{\pi^{2 s} {\cal F}_{2 m}(s)}{\cal C}(0,1;s) {\cal C}(1,4 m;s). \label{mz13}
\end{equation}
It obeys the functional equation
\begin{equation}
\Delta_3(2, 2m;s)={\cal F}_{2 m}(1-s)\Delta_3(2, 2m;1-s) ,
\label{mz14}
\end{equation}
and on the critical line $s=1/2+\ri t$,
\begin{equation}
\Delta_3(2,2m;s)=[1-\ri \tan(\phi_{2 m}(s))][|\widetilde{\cal C}(2,2
m;s)|^2-|\widetilde{\cal S}(2,2 m;s)|^2], \label{mz15}
\end{equation}
and
\begin{equation}
\Imag [\widetilde{\cal C}(2,2 m;1/2+\ri t)]\equiv \Imag [\widetilde{\cal S}(2,2 m;1/2+\ri t)].
\label{mz16}
\end{equation}
We note from (\ref{mz15}) that
\begin{equation}
\Imag \Delta_3(2,2m;\frac{1}{2}+\ri t)=-\tan (\phi_{2 m,c} t) \Real
\Delta_3(2,2m;\frac{1}{2}+\ri t), \label{mz17}
\end{equation}
using the notation $\phi_{ 2m}(1/2+\ri t)=\phi_{2 m,c} (t)$, a real-valued function. We can take the derivative with respect to $t$ of
(\ref{mz17}), to obtain:
\begin{eqnarray}
\frac{\partial}{\partial t}\Imag \Delta_3(2,2 m;\frac{1}{2}+\ri t)+ \tan \phi_{2 m,c}(t)\frac{\partial}{\partial t}\Real \Delta_3(2,2 m;\frac{1}{2}+\ri t) \nonumber \\
=-\frac{\phi_{2 m,c}'(t)}{\cos^2(\phi_{2 m,c}(t))} \Real \Delta_3(2,2 m;\frac{1}{2}+\ri t) .
\label{mz18}
\end{eqnarray}
We use the Cauchy-Riemann equations in (\ref{mz18}), to obtain
\begin{eqnarray}
\frac{\partial}{\partial \sigma}\Real \Delta_3(2,2 m;\frac{1}{2}+\ri t)+ \tan \phi_{2 m,c}(t)\frac{\partial}{\partial t}\Real \Delta_3(2,2 m;\frac{1}{2}+\ri t) \nonumber \\
=-\frac{\phi_{2 m,c}'(t)}{\cos^2(\phi_{2 m,c}(t))} \Real \Delta_3(2,2 m;\frac{1}{2}+\ri t) ,
\label{mz19}
\end{eqnarray}
and
\begin{eqnarray}
\frac{\partial}{\partial t}\Imag \Delta_3(2,2 m;\frac{1}{2}+\ri t)- \tan \phi_{2 m,c}(t)\frac{\partial}{\partial \sigma}\Imag \Delta_3(2,2 m;\frac{1}{2}+\ri t) \nonumber \\
=-\frac{\phi_{2 m,c}'(t)}{\cos^2(\phi_{2 m,c}(t))} \Real \Delta_3(2,2 m;\frac{1}{2}+\ri t) .
\label{mz20}
\end{eqnarray}
The equation (\ref{mz19}) indicates that, at points where contours of $\Real \Delta_3(2,2 m; \sigma+\ri t)=0$ intersect the critical line,
their tangent vector is given by $(1,\tan \phi_{2 m,c}(t))$, provided $\partial \Delta_3(2,2 m;\frac{1}{2}+\ri t)/\partial t\neq 0$. (This requirement
is that the left-hand side of (\ref{mz19}) can be interpreted as the scalar product of the tangent vector and the gradient of
$\Real \Delta_3$, with the latter having a well-defined direction.) The equation (\ref{mz20}) indicates that the tangent vectors
for the contours of $\Imag \Delta_3=0$ at the critical line are given by $(-\tan \phi_{2 m,c}(t),1)$,i.e. they are at right angles to those for the
real part.

We note that, for $|s|>>1$, or, more strictly, for $|s|>>4 m^2$,
\begin{equation}
{\cal F}_{2 m}(s)\simeq 1-\frac{4 m^2}{s-1/2}+\frac{8
m^4}{(s-1/2)^2},~~ \phi_{2 m}(s)\simeq \frac{2 m^2 i}{s-1/2}+\frac{i
m^2(8 m^2-1)}{6(s-1/2)^3}. \label{mz21}
\end{equation}
Thus, for $|t|>>1$, (\ref{mz19}) takes the approximate form
\begin{equation}
\frac{\partial}{\partial \sigma}\Real \Delta_3(2,2 m;\frac{1}{2}+\ri t)+ \frac{2 m^2}{t}\frac{\partial}{\partial t}\Real \Delta_3(2,2 m;\frac{1}{2}+\ri t)
=\frac{2 m^2}{t^2} \Real \Delta_3(2,2 m;\frac{1}{2}+\ri t) .
\label{mz22}
\end{equation}
This shows that, as $t$ increases, the  contours of $\Real
\Delta_3(2,2 m; \sigma+\ri t)=0$ strike the critical line at ever
flatter angles (although the angle increases as $m$ increases). The
direction of the intersection with the critical line is unique for
every point where $\partial \Delta_3(2,2 m;\frac{1}{2}+\ri
t)/\partial t\neq 0$.
% and, from the smoothness of the variation of
% the tangent vector
% with $t$, we see that even where
% that derivative is zero, there can only be one equivalue contour passing through the point.
%This demonstrates that all zeros
%of $\Real \Delta_3(2,2 m; \sigma+\ri t)$ on the critical line are first-order.
Similar remarks apply to $\Imag \Delta_3(2,2 m; \sigma+\ri t)=0$,
where the equivalue contours cut the critical line at a direction
given by the tangent vector $(-2 m^2/t,1)$, and thus their gradient
at the point of intersection increases with $t$.

\section{Equivalue contours of $\Real \Delta_3$  and $\Imag \Delta_3$}

\begin{figure}
\begin{center}
\includegraphics[width=6cm]{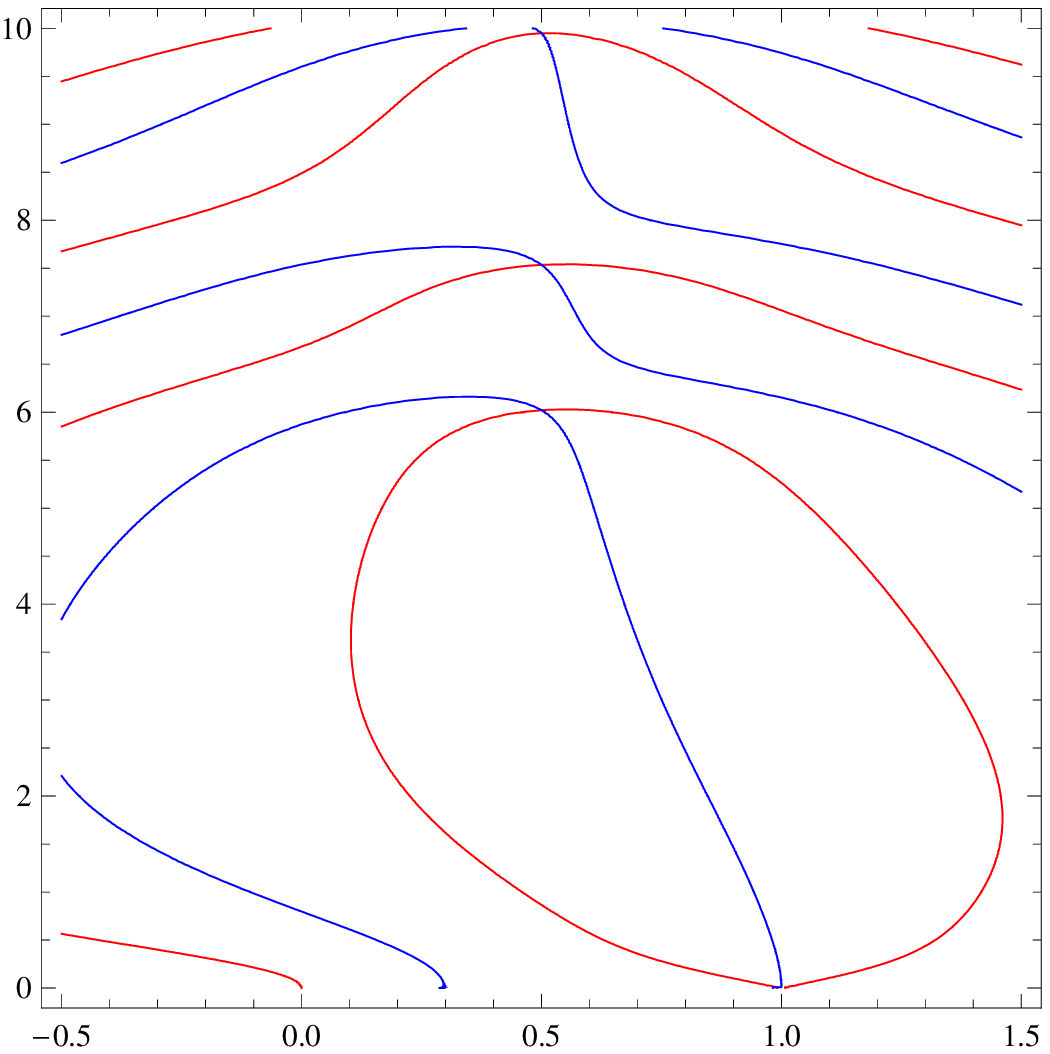}~~~
\includegraphics[width=6cm]{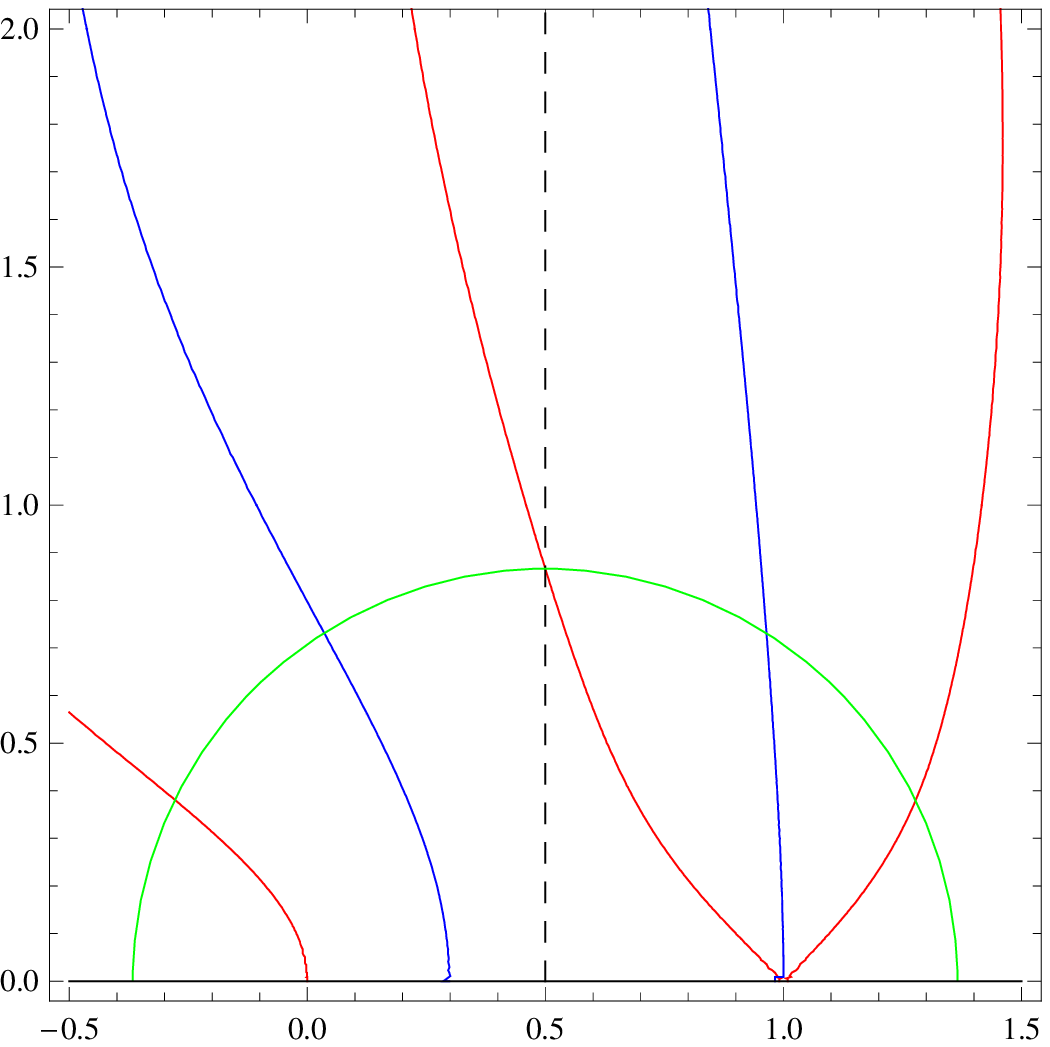}
\end{center}
\caption{Null contours of the real part (red) and imaginary part
(blue) of $\Delta_3(2,2;\sigma+\ri t)$, with (left)$\sigma \in [-0.5, 1.5]$,
and $t \in [0, 10]$. On the right we show
part of the left panel, with now $t \in [0.0, 2.0]$, to
give an undistorted geometric image. On the green circle, ${\cal
F}_2(s)$ is real and negative.} \label{fig1}
\end{figure}

\begin{figure}
\begin{center}
\includegraphics[width=12cm]{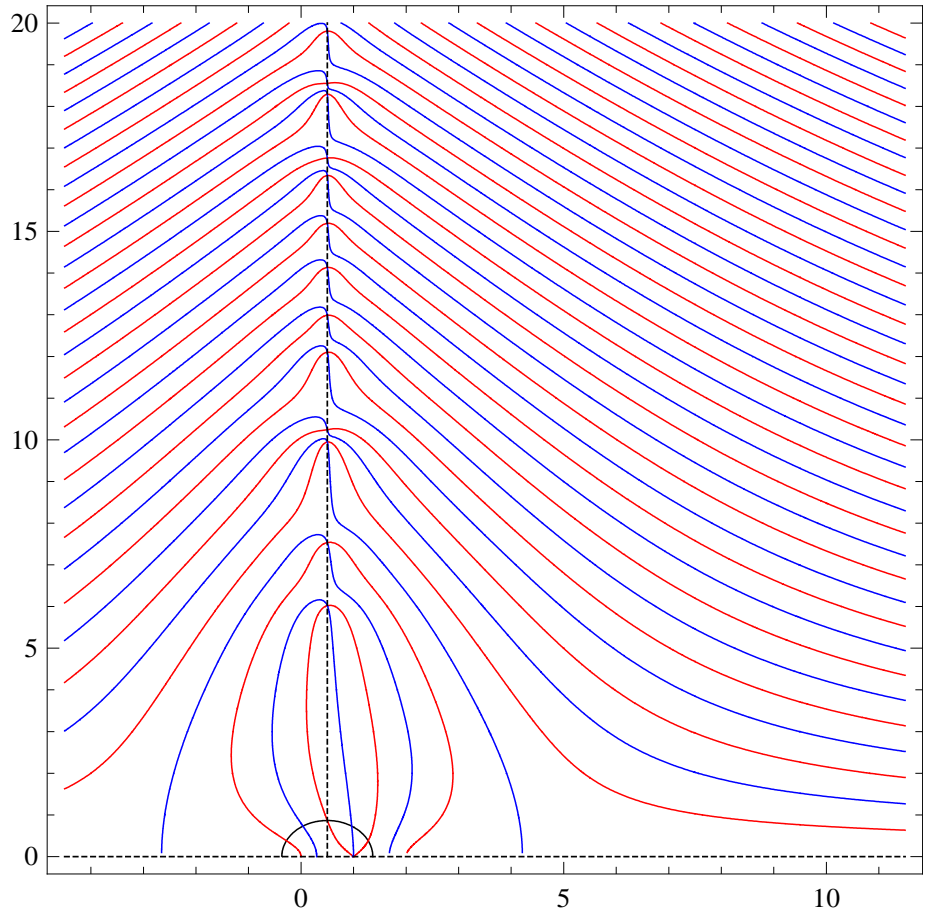}
\end{center}
\caption{Null contours of the real part (red) and imaginary part
(blue) of $\Delta_3(2,2;\sigma+\ri t)$, with $\sigma \in [-4.5, 11.5]$,
and $t \in [0.1, 20]$.} \label{fig2}
\end{figure}

In Figs. \ref{fig1}-\ref{fig3} we show contours on which the real
and imaginary parts of $\Delta_3(2,2;s)$ are zero. Fig. \ref{fig1}
gives some detail of the region near $\sigma=0.5$ for $t\in
[0.1,10]$, and a further more detailed region on an undistorted
geometric scale (including the semi-circle on which ${\cal F}_2(s)$
is real).  Fig. \ref{fig2}  gives a more global view of null
contours for $t$ ranging up to 20. Fig. \ref{fig3} gives the detail
of the null contours for $t$ near two values at which contours
nearly touch. We introduce the notation $14$ for a zero on
$\sigma=1/2$ of ${\cal C}(1,4;s)$, $+1$ for a zero of $\zeta (s)$ and
$-4$ for a zero of $L_{-4}(s)$. Then the zeros evident in
Fig.\ref{fig2} are categorized as:
$-4,14,14,-4,14,-4,+1,14,-4,14,-4,14,14$.

The null contours  have a number of interesting properties.
Firstly, we can see that contours for the real part do indeed strike the critical line at a small angle to the $\sigma$ axis, which decreases
as $t$ increases, while the contours for the imaginary part strike the critical line almost vertically. Secondly, the contours
for the real and imaginary
parts intersect the critical line simultaneously, except for one point. This is at $(1/2,\sqrt{3}/2)$, where, as remarked in (I),
$\phi_2(1/2+ i\sqrt{3/2})=\pi/2$, which permits $\Real \Delta_3(2,2;s)$ to be zero, while $\Imag \Delta_3(2,2;s)$ is non-zero (see (\ref{mz15})).

\begin{figure}
\begin{center}
\includegraphics[width=6cm]{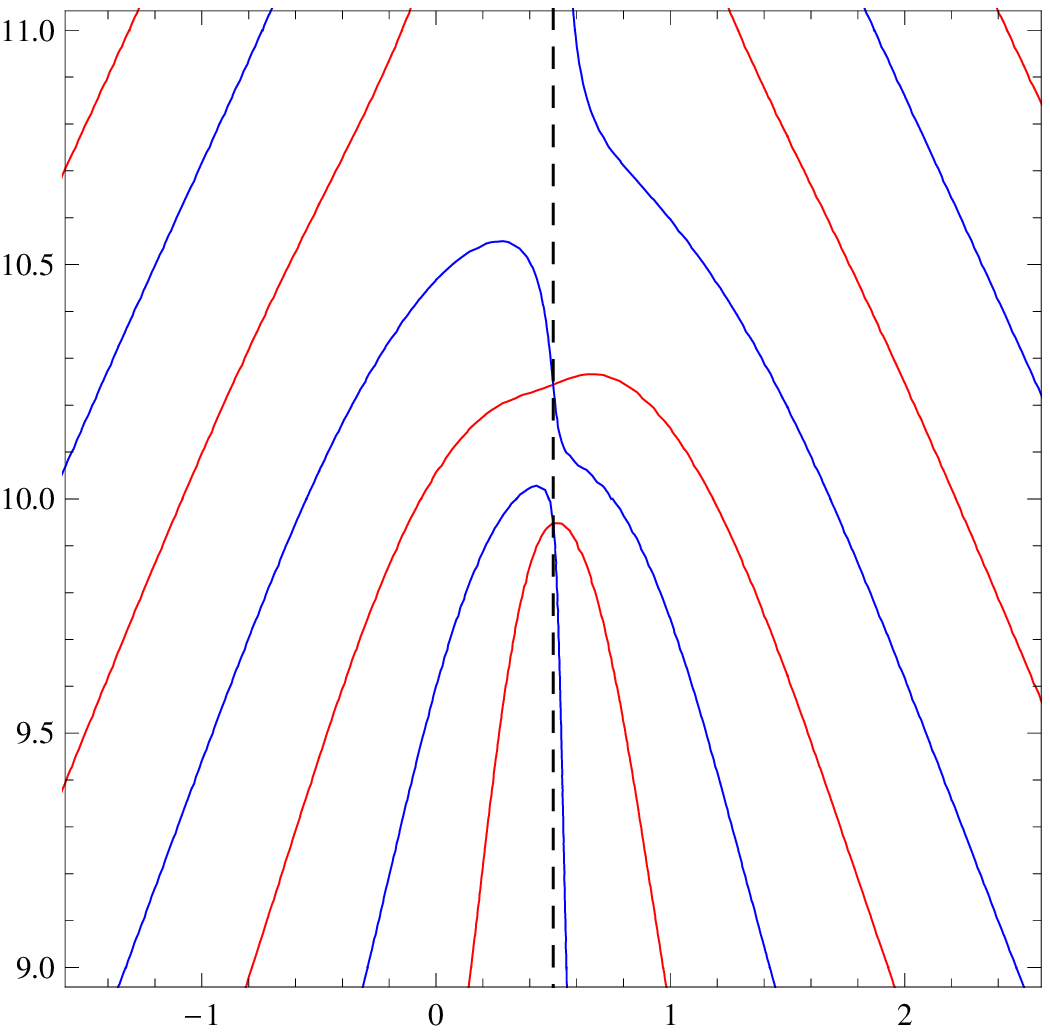}~~~
\includegraphics[width=6cm]{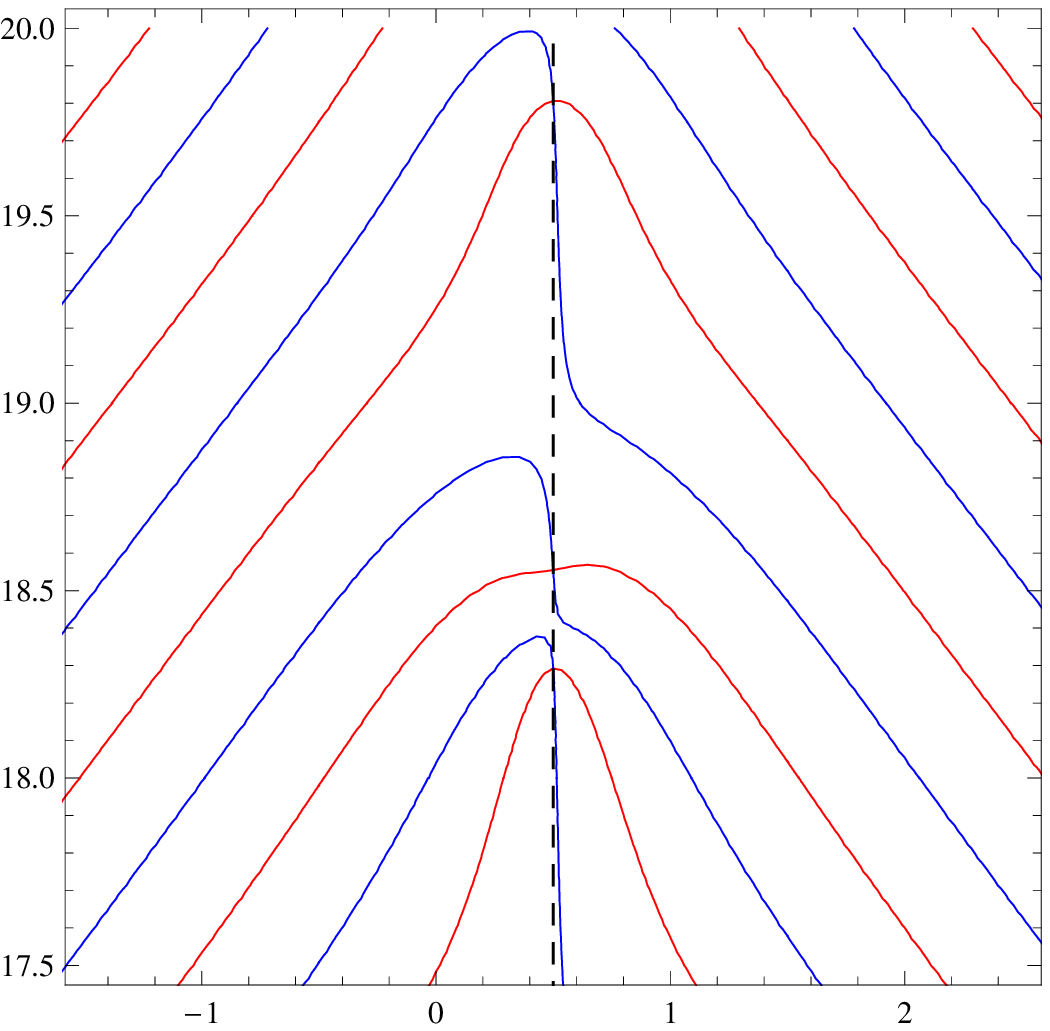}
\end{center}
\caption{Detail of the null contours of the real part (red) and
imaginary part (blue) of $\Delta_3(2,2;\sigma+\ri t)$, for $t$ near
10 (left) and near 18.5 (right).} \label{fig3}
\end{figure}

We can see four null contours for the real part and five null contours for the imaginary part intersecting the axis $t=0$. If we use
the expression
\begin{equation}
\Delta_3(2,2;s)=\frac{\Gamma(s)s (s+1)}{(2-s)(1-s)\pi^{2 s}}[{\cal
C}^2(2,2;s)-{\cal S}^2(2,2;s)], \label{ev1}
\end{equation}
we find that $\Delta_3(2,2;s)$ has a second order pole at $s=1$, and first order poles at $s=0$ and $s=2$. Numerical investigations show
that near these points
\begin{equation}
\Delta_3(2,2;1+\delta)\simeq\frac{-1.59643}{\delta^2},
\label{ev2}
\end{equation}
and
\begin{equation}
\Delta_3(2,2;\delta)\simeq\frac{-0.798212}{\delta},~~\Delta_3(2,2;2+\delta)\simeq\frac{1.16981}{\delta}.
\label{ev3}
\end{equation}
We note that $\Real \Delta_3(2,2;s)$ has zeros at the two first-order poles, where the trajectories approach the poles broadside (i.e., parallel to
the $t$ axis), and two  null trajectories approaching the second-order pole at $45^\circ$ and $135^\circ$ to the $\sigma$ axis.
$\Imag \Delta_3(2,2;s)$ has a zero at the second-order pole, where the null trajectory approaches broadside. The other null trajectory
crossings of the real axis occur at 0.29782, 1.67735, -2.65568 and 4.21422. The last two of these are associated with minima
of $\Real \Delta_3(2,2;\sigma)$.

\begin{figure}
\begin{center}
\includegraphics[width=6cm]{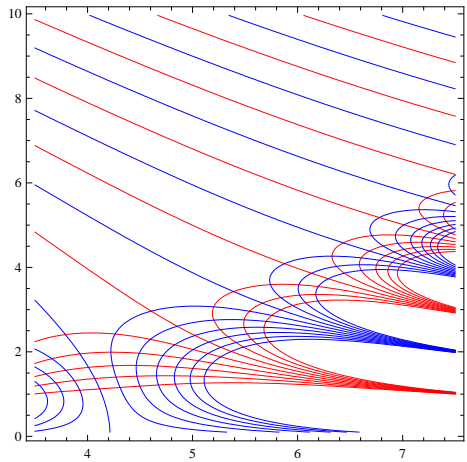}
\end{center}
\caption{Equivalue contours of the real part (red) and imaginary
part (blue) of  the form (\ref{ev5}) for $\Delta_3(2,2;\sigma+\ri
t)$, with $\sigma \in [3.5, 7.5]$, and $t \in [0.1, 10]$.} \label{fig4}
\end{figure}

We show in Fig. \ref{fig4} the behaviour of the null contours for
$\sigma$ large enough  ($\sigma>4$) to enable us to accurately
approximate trigonometric sums by their first few terms:
\begin{equation}
{\cal C}(2,2;s)^2-{\cal S}(2,2;s)^2\simeq 16 (1+\frac{1}{4^s}+\frac{36}{5^{2 s+2}}),
\label{ev4}
\end{equation}
so that
\begin{equation}
\Delta_3(2,2;s)\simeq \left[\frac{16 \Gamma(1-s) \Gamma(s+2) \Gamma(s)}{\Gamma(3-s) \pi^{ 2s}}\right]
(1+\frac{1}{4^s}+\frac{36}{5^{2 s+2}}).
\label{ev5}
\end{equation}
For $\sigma$ beyond 4.21422, we see from Fig. \ref{fig4} that
neither the null trajectories of the real part or the imaginary part
of $\Delta_3(2,2;s)$ can attain the real axis. In Fig. \ref{fig5},
we show for comparison the null trajectories associated with the
prefactor term in square brackets in (\ref{ev5}). It is evident from
a comparison of Figs. \ref{fig4} and \ref{fig5} that the null
trajectories are given accurately by the prefactor in the range of
$\sigma$ shown. Using Stirling's formula (Abramowitz and Stegun
(1972)), we can place the prefactor in exponential form for $|s|$
large, and find the constraint for null trajectories to be
\begin{equation}
\Imag [2 s\log(\frac{s}{\pi})- 2 s-\log(s)+\frac{25}{6 s}+\frac{2}{s^2}]=(n+\frac{1}{2})\pi,~{\rm or}~ n\pi,
\label{ev6}
\end{equation}
where on the right-hand side the first value is for real-part zeros, and the second for imaginary-part zeros ($n$ being an integer).
From (\ref{ev6}) we find that the equivalue contours alternate for $\sigma$ large, with null contours for the real part sandwiched
between those for the imaginary part, and vice versa (with each trajectory corresponding to a first-order zero). Each contour
tends to zero as $1/\log(\sigma/\pi)$ as $\sigma\rightarrow\infty$.
\begin{figure}
\begin{center}
\includegraphics[width=6cm]{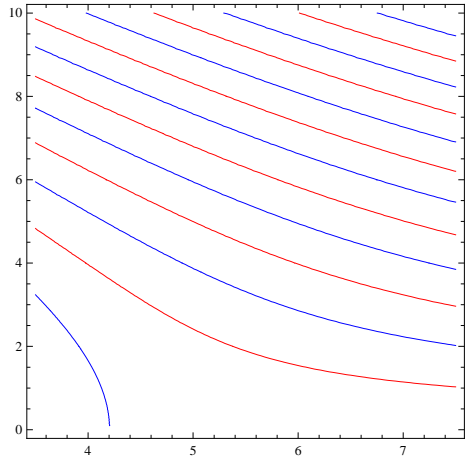}
\end{center}
\caption{Null contours of the real part (red) and imaginary part
(blue) of  the prefactor term in the expression (\ref{ev5}) for
$\Delta_3(2,2;\sigma+\ri t)$, with $\sigma \in [3.5, 7.5]$, and $t \in [0.1, 10]$.} \label{fig5}
\end{figure}

Note that we may use the equation (\ref{mz14}) to deduce relations governing the phase of $\Delta_3(2, 2m;s)$. Dividing (\ref{mz14}) by its
conjugated form, we obtain
\begin{equation}
\arg (\Delta_3(2, 2m;s))+\arg (\Delta_3(2, 2m;1-\bar{s}))=-\arg ({\cal F}_{2 m}(s)).
\label{arg1}
\end{equation}
This agrees with (\ref{mz15}) and (\ref{mz17}) when $\Real (s)=1/2$. Furthermore, on the right-hand side of Fig. \ref{fig2} we know that
$\arg (\Delta_3(2,2;s)$ increases monotonically with $t$, from its value of zero on the real axis. This enables us to assign values
of the constant phase of the null contours there: $\pi/2$, $\pi$, $-\pi/2$, $0$, $\pi/2$, etc. Using (\ref{arg1})
and (\ref{mz21}), we see that the monotonic increase of $\arg (\Delta_3(2,2;s)$ with $t$ on the right in Fig. \ref{fig2}) forces
a monotonic decrease of $\arg (\Delta_3(2,2;s)$ with $t$ on the left. We can then assign the phase values of the null contours there, again
starting from zero on the real axis: $-\pi/2$, $-\pi$, $\pi/2$, $0$, $-\pi/2$, etc.

We can use this information to understand the behaviour of the null
contours shown in Fig. \ref{fig2}. As the null contours of $\Real
(\Delta_3(2, 2m;s))$ pass through a simple zero, their phase must
change by $\pi$, from $\pi/2$ to $-\pi/2$, or vice versa. However,
we note that the relevant contours on the left in Fig. \ref{fig2}
which are almost symmetric to those on the right have opposite phase
values. This means the null contours of $\Real (\Delta_3(2, 2m;s))$
vary smoothly as they cross the critical line. For the null contours
of $\Imag (\Delta_3(2, 2m;s))$ the situation is different: the phase
change of $\pi$ forces them to avoid their almost symmetric
counterpart on the left and "jump up" a contour as they  pass
through a zero on the critical line. These remarks are in accord
with the derivative estimates at the critical line (see equation
(\ref{mz22}) and subsequent discussion).

\section{Some properties of zeros}
The properties of the angular sums given in the preceding sections
enable simple proofs to be given of some properties concerning the
location of zeros of angular lattice sums. These properties form an
interesting counterpoint to those of the zeros of Epstein zeta
functions (see Bogomolny and Lebouef (1994) and references [26-30]
of that paper). The Epstein zeta functions  are characterized as
having an infinite number of zeros lying on the critical line, but
with many zeros lying off that line and with almost all zeros lying
on the critical line or in its immediate neighbourhood. We have
already seen in the example of Section 2 that this third property of
Epstein zeta function zeros is not shared by those of the
trigonometric double sum $C(2 m,1;s)$ for large $m$.

\begin{theorem}
The trigonometric sums $\widetilde{C}(2,2 m;s)$ and
$\widetilde{S}(2,2 m;s)$ have no zeros for $s$ on the critical line
in the asymptotic region $t>>4 m^2$ which are not zeros of both, and
these will then be simultaneous zeros of $ C(0,1;s)$ and $C(1,4
m;s)$. \label{thm1}
\end{theorem}
\begin{proof}
We combine equations (\ref{mz6}) and (\ref{mz6a}), to obtain
\begin{equation}
\widetilde{C}(2,2 m;s)+\widetilde{S}(2,2 m;s)=\frac{1}{\sqrt{{\cal
F}_{2 m}(s)}}\left[\frac{\Gamma (s)}{\pi^s}C(0,1;s)\right].
\label{np1}
\end{equation}
As the term in square brackets is real on the critical line, from
the functional equation (\ref{mz4}), we see that
\begin{equation}
\arg[\widetilde{C}(2,2 m;\frac{1}{2}+ i t)+\widetilde{S}(2,2
m;\frac{1}{2}+i t)]=-\phi_{2 m;c}(t)+
\left[\begin{array}{c}0\\\pi\end{array}\right].\label{np2}
\end{equation}
The second term on the right-hand side of (\ref{np2}) is chosen
according to the sign of the term in square brackets in (\ref{np1}).
We can also use (\ref{mz16}) to deduce that
\begin{equation}
\arg[\widetilde{C}(2,2 m;\frac{1}{2}+ i t)-\widetilde{S}(2,2
m;\frac{1}{2}+i t)]= \left[\begin{array}{c}0\\\pi\end{array}\right].
\label{np3}
\end{equation}
Note that the term on the right-hand side in (\ref{np3})
incorporates the phase of the real function on the left-hand side,
and is not necessarily the same as the second term on the right-hand
side of (\ref{np2}).

If we suppose that $\widetilde{S}(2,2 m;\frac{1}{2}+i t)=0$, we
require from (\ref{np2}) and (\ref{np3}) that $\phi_{2 m;c}(t)=n
\pi$ for some integer $n$. We know from (\ref{mz21}) that this
cannot occur. Exactly the same argument applies if
$\widetilde{C}(2,2 m;\frac{1}{2}+i t)=0$, so neither function can be
separately zero on the critical line in the asymptotic region. If
both are zero then, trivially, $C(0,1;s)=0$ and $C(1,4 m;s)=0$.
\end{proof}

We now consider the phases of  $\widetilde{S}(2,2 m;\frac{1}{2}+i
t)$ and $\widetilde{C}(2,2 m;\frac{1}{2}+i t)$, particularly in the
neighbourhood of zeros $C(0,1;s)$ or $C(1,4 m;s)$. For brevity, we
will adopt the following notations:
\begin{equation}
|\widetilde{C}(2,2 m;\frac{1}{2}+i t)|=|\widetilde{C}|,
~~\arg\widetilde{C}(2,2 m;\frac{1}{2}+i t)= \Theta_c, \label{np4}
\end{equation}
and
\begin{equation}
|\widetilde{S}(2,2 m;\frac{1}{2}+i t)|=|\widetilde{S}|,
~~\arg\widetilde{S}(2,2 m;\frac{1}{2}+i t)= \Theta_s. \label{np5}
\end{equation}
Then on the critical line, we have for  $\widetilde{C}$ and
$\widetilde{S}$ that
\begin{equation}
\widetilde{C}(2,2 m;\frac{1}{2}+i t)= |\widetilde{C}| (\cos
\Theta_c+i \sin \Theta_c),~~\widetilde{S}(2,2 m;\frac{1}{2}+i t)=
|\widetilde{S}| (\cos \Theta_s+i \sin \Theta_s), \label{np6}
\end{equation}
and
\begin{equation}
|\widetilde{C}| \sin \Theta_c= |\widetilde{S}| \sin \Theta_s .
\label{np7}
\end{equation}
The relationships between these phase angles are exemplified in Fig. \ref{nfig1}, where the
variations with $t$ of the functions $\cot(\Theta_c)$ and
$\cot(\Theta_s)$ are shown for the case $m=1$.

\begin{theorem}
The phase angles $\Theta_c$ and  $\Theta_s$ obey the following equations
on the critical line
\begin{equation}
\cot (\Theta_c)+ \cot (\Theta_s)=-2 \cot (\phi_{2
m,c}(t) ,\label{np8a}
\end{equation}
and
\begin{equation}
|\cot (\Theta_c)- \cot (\Theta_s)|=\frac{2|C(1,4 m;s)|}{|C(0,1;s)||\sin (\phi_{2 m,c}(t)|} .\label{np8aa}
\end{equation}
\label{newthm1}
\end{theorem}
\begin{proof}
We prove these relationships using the diagram of Fig. \ref{nfig2}, in which complex quantities are represented by vectors.
For the case shown, the angle $\Theta_s$ runs from the extended line $OA$ round to $AD$, while $\Theta_c$ runs from  line $OA$ round to $OM$, with the
former angle exceeding the latter.
The vectors representing $\widetilde{C}-\widetilde{S}$ and $\widetilde{C}+\widetilde{S}$ are respectively $OA$ and $OD$, while the angle between
$OA$ and $OD$ is $\pi-\phi_{2m,c}(t)$.

We apply the sine rule to the triangle $OAB$, to give
\begin{equation}
\frac{2|\widetilde{C}|}{\sin (\phi_{2 m,c}(t))}=\frac{|\widetilde{C}+\widetilde{S}|}{\sin (\Theta_c)}=
\frac{|\widetilde{C}-\widetilde{S}|}{\sin (\Theta_c+\phi_{2 m,c}(t))}.
\label{sine1}
\end{equation}
We next apply the sine rule to the triangle $OAD$, to give
\begin{equation}
\frac{2|\widetilde{S}|}{\sin (\phi_{2 m,c}(t))}=\frac{|\widetilde{C}+\widetilde{S}|}{\sin (\Theta_s)}=
\frac{|\widetilde{C}-\widetilde{S}|}{\sin (\Theta_c+\phi_{2 m,c}(t)-\pi)}.
\label{sine2}
\end{equation}
Using the first two elements of (\ref{sine1}) and (\ref{sine2}) we recover the result that
$|\widetilde{C}|/|\widetilde{S}|=\sin (\Theta_s)/\sin (\Theta_c)$. Using the second and third elements, we find
\begin{equation}
\cos(\phi_{2 m,c}(t))+\cot (\Theta_c) \sin(\phi_{2 m,c}(t))=\frac{|\widetilde{C}-\widetilde{S}|}{|\widetilde{C}-\widetilde{S}|},
\label{sine3}
\end{equation}
and
\begin{equation}
-\cos(\phi_{2 m,c}(t))-\cot (\Theta_s) \sin(\phi_{2 m,c}(t))=\frac{|\widetilde{C}-\widetilde{S}|}{|\widetilde{C}-\widetilde{S}|}.
\label{sine4}
\end{equation}
Subtracting (\ref{sine3}) and (\ref{sine4}), we obtain (\ref{np8a}), while adding them and using
\begin{equation}
\frac{|\widetilde{C}-\widetilde{S}|}{|\widetilde{C}+\widetilde{S}|}=\frac{|C(1,4 m;s)|}{|C(0,1;s)|}
\label{sine5}
\end{equation}
we obtain (\ref{np8aa}) for the particular case of $\Theta_s>\Theta_c$ and $\sin(\phi_{2 m,c}(t))>0$. The form given for (\ref{np8aa}) is true
whatever the inequality between $\Theta_s$ and $\Theta_c$, and whatever the sign of $\sin(\phi_{2 m,c}(t))$.
\end{proof}

The curves of Fig. \ref{nfig1} giving $\cot (\Theta_c)$ and $ \cot
(\Theta_s)$ on the critical line as a function of $t$ for $m=1$ illustrate the relationship (\ref{np8a}).
One may be obtained from the other by reflection in the line giving the
variation of $-\cot (\phi_{2 m,c}(t)$. They intersect when $\Theta_c=\Theta_s$, at a zero  of $C(1,4 m;\frac{1}{2}+i t)$.
The curves of  $\cot
(\Theta_c)$ and $ \cot (\Theta_s)$ also go off to infinity at zeros
of $C(0,1;\frac{1}{2}+i t)$, with their movement being in opposite
directions in accord with (\ref{np8a}). In Fig. \ref{nfig1} we see
examples of there being none, one or two zeros $C(1,4
m;\frac{1}{2}+i t)$ between successive zeros of $C(0,1;\frac{1}{2}+i
t)$.

 Using (\ref{np6}) and (\ref{np7}), we
can rewrite (\ref{np2}) as
\begin{equation}
\frac{ 2 \sin \Theta_c   \sin \Theta_s
}{\sin(\Theta_c+\Theta_s)}=-\tan \phi_{2 m,c}(t).\label{np8}
\end{equation}
With a prime denoting $t$ derivatives, we obtain by differentiating
(\ref{np8})
\begin{eqnarray}
(2\cos \Theta_c \sin \Theta_s) \Theta_c '+  (2\sin \Theta_c \cos
\Theta_s) \Theta_s '=& & \nonumber \\
-\frac{\phi_{2 m,c}'(t)}{\cos^2 \phi_{2 m,c}(t)} \sin
(\Theta_s+\Theta_c) -\tan \phi_{2 m,c}(t) \cos( \Theta_s+\Theta_c)
(\Theta_s'+\Theta_c'). \label{np9}
\end{eqnarray}

We now consider the values of $\Theta_s$, $\Theta_c$ and their $t$
derivatives in the neighbourhood of zeros of $C(1,4 m;\frac{1}{2}+i
t)$ or $C(0,1 ;\frac{1}{2}+i t)$, which are not zeros of both
functions. Close to a zero of $C(1,4 m;\frac{1}{2}+i t)$, we require
\begin{equation}
|\widetilde{C}| \cos \Theta_c- |\widetilde{S}| \cos \Theta_s
\rightarrow 0,~~ |\widetilde{C}|\rightarrow  |\widetilde{S}|,
~~\Theta_c\rightarrow \Theta_s,  \label{np10}
\end{equation}
where $|\widetilde{C}|\neq 0$, $|\widetilde{S}|\neq 0$ at the zero
by Theorem \ref{thm1}. Using the last of (\ref{np10}) in (\ref{np8})
and (\ref{np9}), we find at the zero of $C(1,4 m;\frac{1}{2}+i t)$
that
\begin{equation}
\Theta_s=\Theta_c=-\phi_{2m,c}(t)+\left[\begin{array}{c}0\\\pi\end{array}\right],~~
\Theta_c'+\Theta_s'=-2\phi_{2 m,c}'(t). \label{np11}
\end{equation}

Close to a zero of $C(0,1;\frac{1}{2}+i t)$, we require
\begin{equation}
|\widetilde{C}| \cos \Theta_c+  |\widetilde{S}| \cos \Theta_s
\rightarrow 0,  |\widetilde{C}| \sin \Theta_c+  |\widetilde{S}| \sin
\Theta_s \rightarrow 0,~~ |\widetilde{C}|\rightarrow
|\widetilde{S}|. \label{np12}
\end{equation}
From (\ref{np12}), at the zero
\begin{equation}
\Theta_c=\pi+\Theta_s=
\left[\begin{array}{c}0\\\pi\end{array}\right],~~
\Theta_s'=-\Theta_c'.\label{np13}
\end{equation}

\begin{figure}
\begin{center}
\includegraphics[width=6cm]{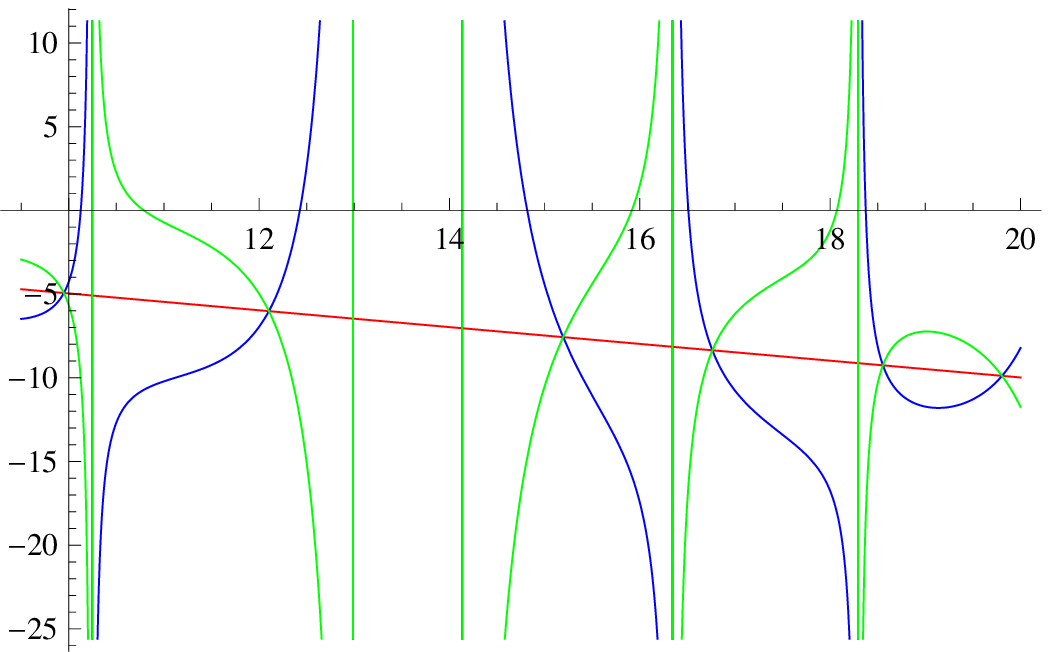}~~\includegraphics[width=6cm]{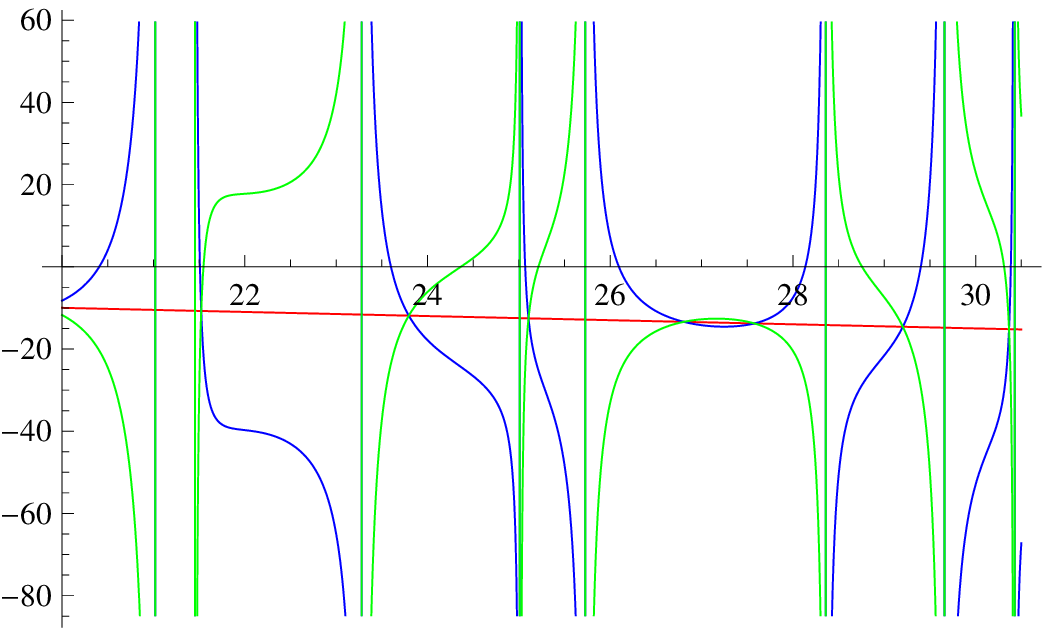}~~
\end{center}
\caption{Variation of the quantities $-\cot(\phi_{2,c}(t))$ (red),
$\cot(\Theta_c (t))$ (blue) and  $\cot(\Theta_s(t))$ (green) with $t
\in [9.5, 20]$ (left) and $[20, 30.5]$ (right), where the angles $\Theta_c$ and $\Theta_s$
pertain to $\widetilde{C}(2,2;s)$ and $\widetilde{S}(2,2;s)$ .}
\label{nfig1}
\end{figure}

\begin{figure}
\begin{center}
\includegraphics[width=8cm]{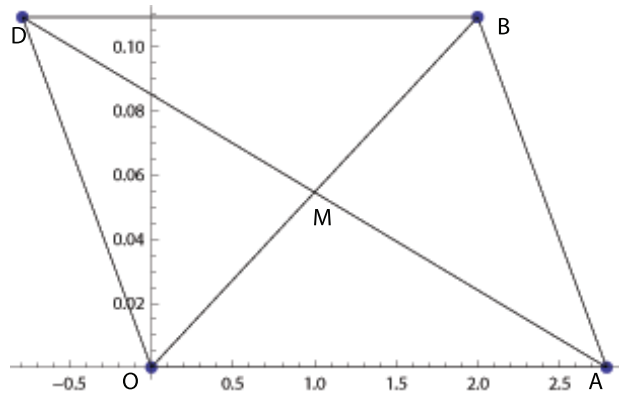}
\end{center}
\caption{The complex quantities $\widetilde{C}(2,2;s)$ ($OM$) and $
\widetilde{S}(2,2;s)$  ($AM$) represented as vectors, together with
their sum ($OD$) and difference ($OA$), with the last two
respectively proportional to $C(0,1;s)$ and $C(1,4;s)$.}
\label{nfig2}
\end{figure}

The next result relates to  the independence of the functions
$\Delta_3 (2,2 m;s)$ on the critical line, again in the asymptotic
region.
\begin{theorem}
If $s$ is on the critical line in the asymptotic region $t>>4 l^2$
and $\alpha_l$, $\beta$ are non-zero reals, then
\begin{equation}
\alpha_l C(1,4 l;s)+\beta C(0,1;s)\neq 0,  ~{\rm unless}~C(1,4
l;s)=0=C(0,1;s).\label{zprf4}
\end{equation}
If $s$ is on the critical line in the asymptotic region $t>>4 \max
(l^2,m^2)$ and $\alpha_l$, $\alpha_m$ are non-zero reals with $l\neq
m$, then
\begin{equation}
\alpha_l C(1,4 l;s)+\alpha_m C(1,4 m;s)\neq 0 ~{\rm unless}~C(1,4
l;s)=0=C(1,4 m;s). \label{zprf5}
\end{equation}
\label{thm2}
\end{theorem}
\begin{proof}
We suppose
\begin{equation}
\alpha_l C(1,4 l;s)+\beta C(0,1;s)= 0 ~{\rm with}~C(0,1;s)\neq 0.
\label{zprf6}
\end{equation}
Then
\begin{equation}
\alpha_l \frac{C(1,4 l;s)}{C(0,1;s)} =-\beta, ~~ \arg
\left[\frac{C(1,4 l;s)}{C(0,1;s)} \right]= m\pi , \label{zprf7}
\end{equation}
for $m$ an integer.

We now use the functional equations (\ref{mz4}) for $C(0,1;s)$ and
$C(1,4 m;s)$, which on the critical line with $1-s=\overline{s}$
enable their arguments to be deduced:
\begin{equation}
\arg C(0,1;s)=- [\arg \pi^s +\arg \Gamma (s) ] ~,~ \arg C(1,4 m;s)=-
[\arg \pi^s +\arg \Gamma (s+2 m) ]. \label{zprf2}
\end{equation}
Using (\ref{zprf2}) we arrive at the estimate for $s$ on the
critical line in the asymptotic region
\begin{equation}
\arg C(0,1;s)- \arg C(1,4 m;s)=m \pi -\frac{2
m^2}{t}+O(\frac{1}{t^2}).\label{zprf3}
\end{equation}

Using (\ref{zprf3}), (\ref{zprf7}) requires
\begin{equation}
\frac{2 l^2}{t}+O(\frac{1}{t^2})-l \pi=m\pi, \label{zprf8}
\end{equation}
which is not possible. This proves (\ref{zprf4}).

Consider next (\ref{zprf5}), and suppose
\begin{equation}
\alpha_l C(1,4 l;s)+\alpha_m C(1,4 m;s)= 0 ~{\rm with }~C(1,4
m;s)\neq 0. \label{zprf9}
\end{equation}
Then this requires
\begin{equation}
\arg \left[ \frac{C(1,4 l;s)}{C(1,4 m;s)}\right]=p \pi,
\label{zprf10}
\end{equation}
for an integer $p$. However, from (\ref{zprf3}),
\begin{equation}
\arg C(1,4 l;s)-\arg C(1,4 m;s)=(m-l)\pi+\frac{2
(l^2-m^2)}{t}+O(\frac{1}{t^2}). \label{zprf11}
\end{equation}
Hence (\ref{zprf9}) requires
\begin{equation}
\frac{2 (l^2-m^2)}{t}+O(\frac{1}{t^2})=(p+l-m)\pi, \label{zprf12}
\end{equation}
which is not possible if $l\neq m$. This proves (\ref{zprf5}).
\end{proof}
\begin{corollary}
The sums ${\cal C}(2,2,s)$, ${\cal S}(2,2,s)$ and ${\cal C}(4,1;s)$
have no zeros $s$ on the critical line in the asymptotic region
$t>>4$ which are not zeros of both $C(0,1;s)$ and ${\cal C}(1,4;s)$.
\label{coroll3}
\end{corollary}
\begin{proof}
We apply Theorem \ref{thm2} and equation (\ref{zprf4}). Now the sums
${\cal C}(2,2,s)$, ${\cal S}(2,2,s)$ and ${\cal C}(4,1;s)$ belong to
the system of order 4 discussed in Appendix A, which is generated by
linear combinations of ${\cal C}(0,1;s)$ and ${\cal C}(1,4;s)$. All
such combinations are of the form (\ref{zprf4}), and are thus zero
only if both ${\cal C}(0,1;s)=0$ and ${\cal C}(1,4;s)=0$.
\end{proof}

We can also simply establish a corresponding result for the
derivative $\Delta_3'(2,2m;s)$ for $s$ on the critical line.
\begin{theorem}
The derivative function $\Delta_3'(2,2m;s)$ has no zeros for
$s=1/2+i t$ in the asymptotic region $t>>4 m^2$ which are not zeros
of $\Delta_3(2,2 m;s)$ of multiple order. \label{thm4}
\end{theorem}
\begin{proof}
We recall the result (\ref{mz15}) for $s=1/2+i t$:
\begin{equation}
\Delta_3(2,2m;1/2+ i t)=[1-\ri \tan(\phi_{2 m,c
}(t))][|\widetilde{\cal C}(2,2 m;1/2+i t)|^2-|\widetilde{\cal S}(2,2
m;1/2+i t)|^2], \label{mz15aa}
\end{equation}
where $\phi_{2 m,c }(t))$ is a real-valued function. Taking the
derivative of $\Delta_3(2, 2m; 1/2+i t)$ with respect to $1/2+i t$
and separating into real and imaginary parts, we find
\begin{eqnarray}
\Delta_3 '(2,2m;1/2+ i t)&=&-\left\{ \frac{\phi_{2
m,c}'(t)}{\cos^2\phi_{2m,c}(t)}[|\widetilde{\cal C}(2,2 m;1/2+i
t)|^2-|\widetilde{\cal S}(2,2 m;1/2+i t)|^2]\right. \nonumber
\\
& &\left.  +\tan \phi_{2m,c}(t)\frac{d}{dt} [|\widetilde{\cal C}(2,2
m;1/2+i t)|^2-|\widetilde{\cal S}(2,2 m;1/2+i t)|^2] \right\}
\nonumber \\
&& -i \frac{d}{d t} [|\widetilde{\cal C}(2,2 m;1/2+i
t)|^2-|\widetilde{\cal S}(2,2 m;1/2+i t)|^2].\label{mz15ab}
\end{eqnarray}
For $\Delta_3 '(2,2m;1/2+ i t)$ to be zero, we find from the
imaginary part  of (\ref{mz15ab}) that
\begin{equation}
\frac{d}{d t} [|\widetilde{\cal C}(2,2 m;1/2+i
t)|^2-|\widetilde{\cal S}(2,2 m;1/2+i t)|^2]=0. \label{mz15ac}
\end{equation}
We use (\ref{mz15ac}) in the real part of (\ref{mz15ab}), and,
assuming $\tan \phi_{2m,c}(t)$ is non-singular and $\phi_{2
m,c}'(t)\neq 0$, we complement (\ref{mz15ac}) with
\begin{equation}
[|\widetilde{\cal C}(2,2 m;1/2+i t)|^2-|\widetilde{\cal S}(2,2
m;1/2+i t)|^2]=0. \label{mz15ad}
\end{equation}
Since both the assumptions we have just mentioned are true in the
asymptotic region, we have proved the theorem.
\end{proof}

\begin{figure}
\begin{center}
\includegraphics[width=6cm]{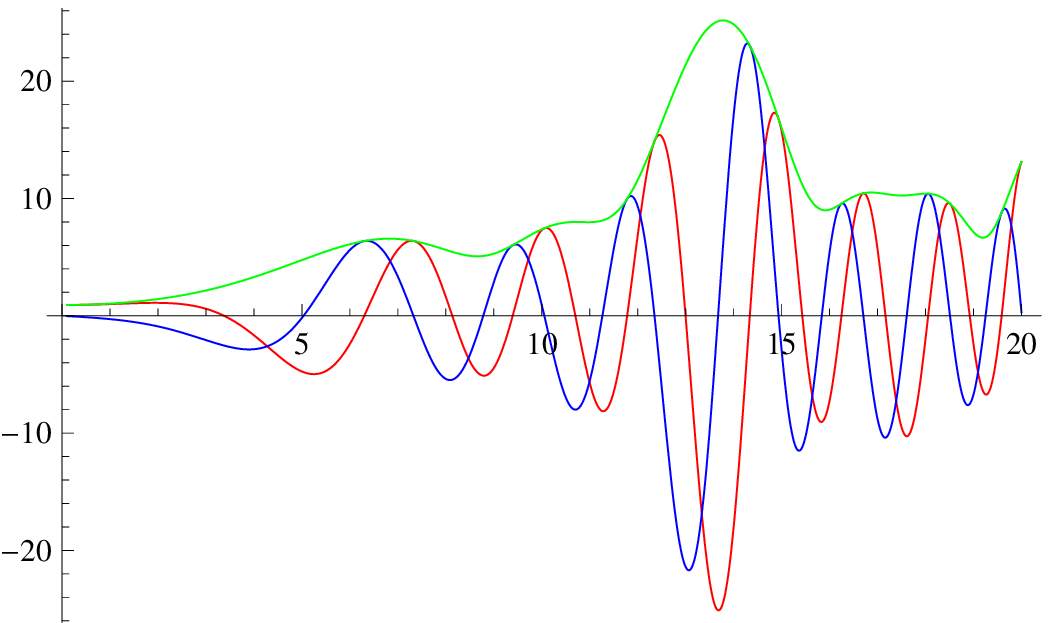}~~\includegraphics[width=6cm]{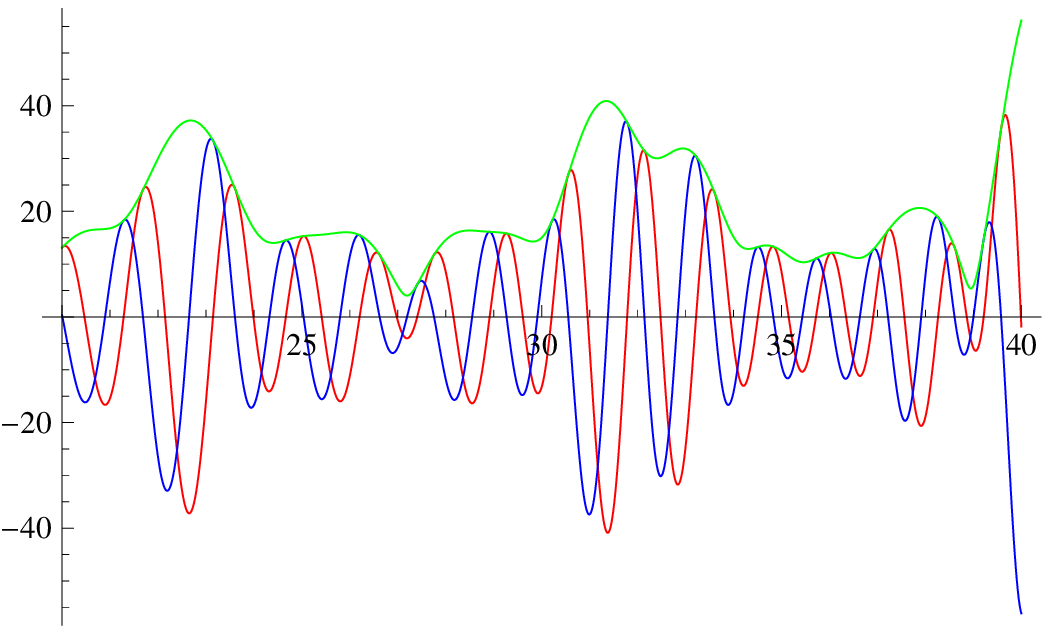}~~
\end{center}
\caption{ The real part (red), imaginary part (blue) and modulus
(green) of  $\Delta_3'(2,2;1/2+\ri t)$, with  $t \in [0.1, 20]$
(left) and $[20,40]$ (right).} \label{derivfig}
\end{figure}

We provide plots of $\Delta_3'(2,2;1/2+\ri t)$ in Fig.
\ref{derivfig}. These show that the derivative is non-zero in the
non-asymptotic region as well as in the asymptotic region. The
numerical value of $\Delta_3'(2,2;1/2)$ is 0.918604.

We next consider the properties of trajectories of constant phase
which follow from the assumption that the Riemann Hypothesis holds
for $\Delta_3 (2,2 m;s)$. As can be seen from Fig. \ref{argdelt},
lines of constant phase between zeros of $\Delta_3 (2,2 ;s)$ do not
in general cross the critical line, but asymptote towards it, and
their configuration is arranged about a zero of the derivative
$\partial|\Delta_3(2,2;\sigma +i t)|/\partial t$ or equivalently of
$\partial \arg \Delta_3(2,2;\sigma +i t)/\partial \sigma$.

\begin{figure}
\begin{center}
\includegraphics[width=6cm]{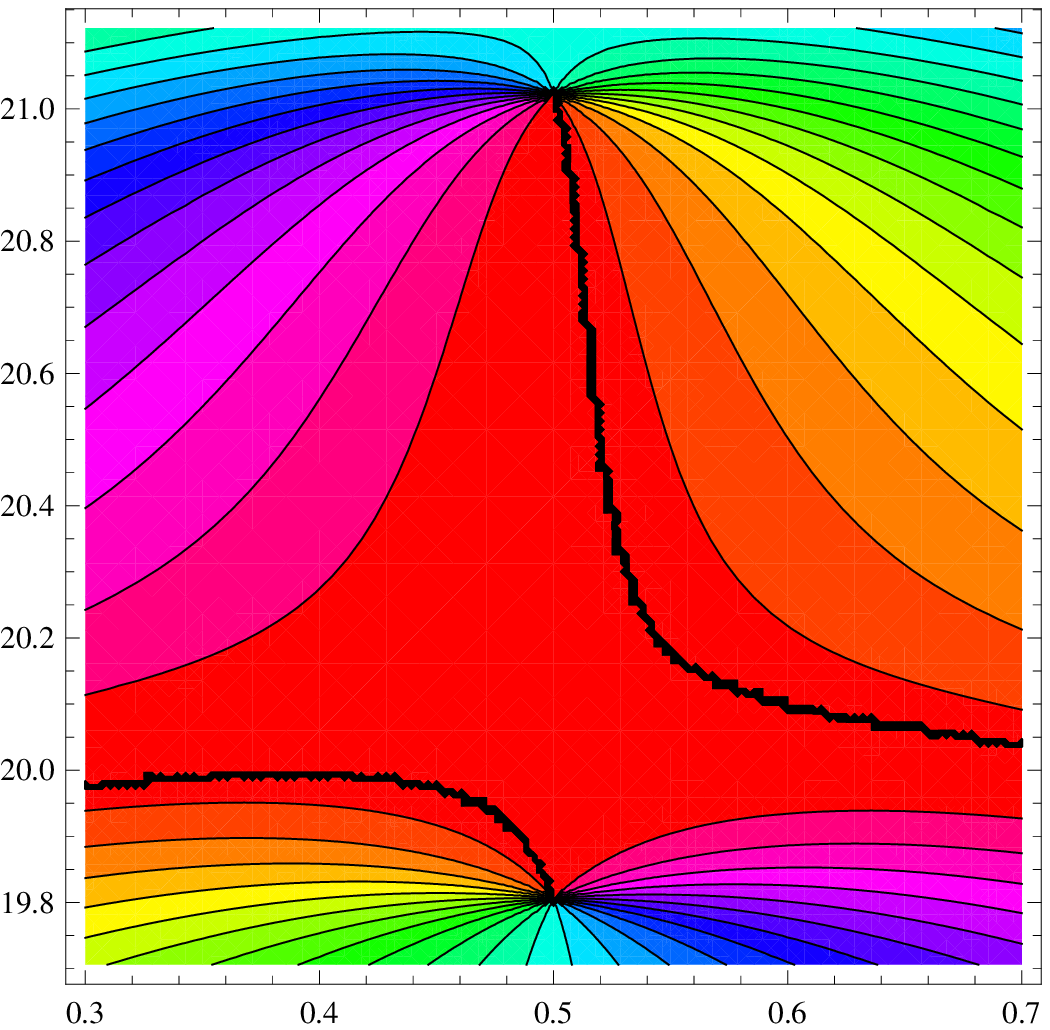}~~\includegraphics[width=6cm]{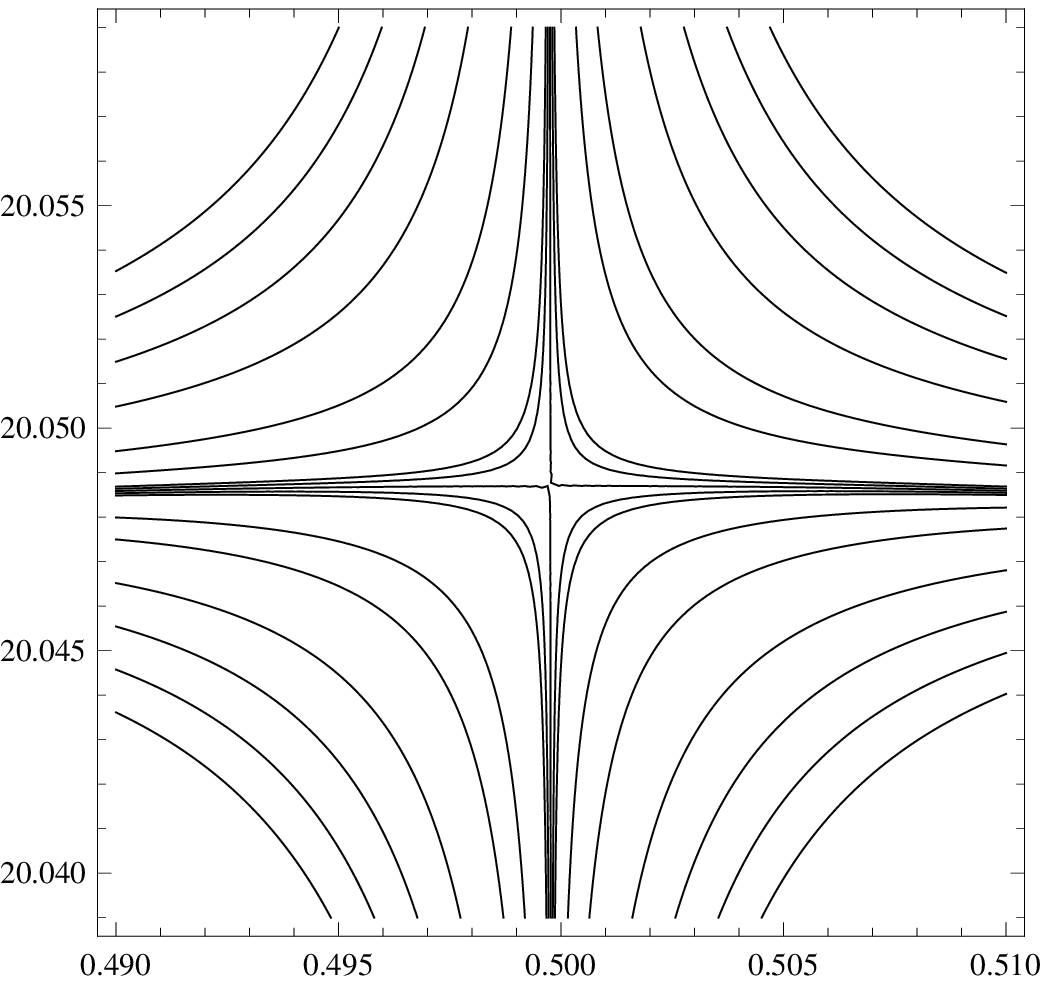}~~
\end{center}
\caption{ Contours of constant phase for $\Delta_3 (2,2;s)$ in the
region around its zero at $s=0.5+i 19.80599$, with a second zero
evident at $s=0.5+ i 21.02204$. The plot at right shows detail
around the zero of the derivative  $\partial
\log|\Delta_3(2,2;\sigma +i t)|/\partial t$. } \label{argdelt}
\end{figure}

\begin{theorem}
Given the Riemann hypothesis holds for $\Delta_3 (2,2 m;s)$, then
lines of constant phase coming from $\sigma=-\infty$ can only cut
the critical line at a zero of $\Delta_3 (2,2 m;s)$ or of $\partial
\arg \Delta_3(2,2;\sigma +i t)/\partial \sigma$. Those lines of
constant phase coming from $\sigma=\infty$  cutting  the critical
line at a point which is not a zero of $\Delta_3 (2,2 m;s)$ or of
$\partial \arg \Delta_3(2,2;\sigma +i t)/\partial \sigma$ must curve
back and pass through a zero of $\Delta_3 (2,2 m;s)$. \label{thm5}
\end{theorem}
\begin{proof}
The assumption of the Riemann hypothesis holding enables us to say
that lines of constant phase coming from $\sigma=-\infty$ do not
intersect before the critical line. Their phase monotonically
decreases as $t$ increases in $\sigma<1/2$, while the phase of
$\Delta (2,2 m;1/2+i t)$ monotonically increases with $t$. Thus,
groups of constant phase lines coming in from $\sigma=-\infty$
cannot cut the critical line, except at a zero of $\Delta (2,2
m;s)$. Isolated trajectories passing through points where $\partial
\arg \Delta(2,2;\sigma +i t)/\partial \sigma=0$ are allowed. Such
special trajectories separate lines of constant phase which curve up
as $\sigma \rightarrow 1/2$ and they move towards a zero of $\Delta
(2,2 m;s)$, from trajectories which curve down as $\sigma
\rightarrow 1/2$.

Now consider sets of lines of constant phase approaching
$\sigma=1/2$ from the right. Their phase increases with $t$, but if
they were able to cross the line $\sigma=1/2$ and progress towards
$\sigma=-\infty$ it would have to decrease with $t$. Thus they must
turn and run alongside $\sigma=1/2$, with lines above the special
trajectory curving up towards a zero of $\Delta (2,2 m;s)$, and
those below curving down towards a zero. Alternatively, they can cut
the critical line at a point which is not a zero of $\Delta_3 (2,2
m;s)$ or of $\partial \arg \Delta(2,2;\sigma +i t)/\partial \sigma$.
Such trajectories supply the lines of constant argument required for
generic points on the critical line, and must return back to cut the
critical line at a zero of $\Delta (2,2 m;s)$. The region where they
cross into $\sigma<1/2$ is bounded on the left by a line of constant
phase connecting a zero of $\Delta_3 (2,2 m;s)$ with an adjacent
zero of $\partial \arg \Delta(2,2;\sigma +i t)/\partial \sigma$.
\end{proof}

Note that the region where zeros of  $\partial \arg
\Delta(2,2;\sigma +i t)/\partial \sigma$ exist is confined to the
neighbourhood of the critical line, since this partial derivative
being zero corresponds to a horizontal segment on a line of constant
phase (see, for example, Fig. \ref{fig2}). However, we may use the
line of constant phase passing through the point on the critical
line where  $\partial \arg \Delta(2,2;\sigma +i t)/\partial
\sigma=0$ as the separator between lines of constant phase going to
the zero above this line from those going to the zero below (given
the Riemann hypothesis is assumed to hold).

\begin{theorem}
Given the Riemann hypothesis holds for $\Delta_3 (2,2 m;s)$, then
there exists one and only one zero of $\partial \arg
\Delta_3(2,2;\sigma +i t)/\partial \sigma$ on the critical line
between two successive zeros of $\Delta_3 (2,2 m;s)$. \label{thm6}
\end{theorem}
\begin{proof}
We consider the analytic function
\begin{equation}
\log \Delta_3(2,2 m;s)=\log |\Delta_3(2,2 m;s)| + i \arg\Delta_3(2,2
m;s). \label{pd1}
\end{equation}

The real part of this function goes to $-\infty$ at any zero of
$\Delta_3(2,2 m;s)$, and increases away from these logarithmic
singularities. It must have at least one turning point between
successive zeros. By the Cauchy-Riemann equations, such a turning
point is a zero of $\partial \arg \Delta_3(2,2;\sigma +i t)/\partial
\sigma$.

Next, suppose there two or more zeros of $\partial \arg
\Delta(2,2;\sigma +i t)/\partial \sigma$ between successive zeros
$s_0=1/2+i t_0$ and $s_1=1/2+ i t_1$. Denote the upper two of such
derivative zeros by $s_*=1/2+i t_*$ and $s_{**}=1/2+i t_{**}$. As we
have seen, each of these has constant phase lines coming from
$\sigma=-\infty$ and passing through it, around which constant phase
trajectories reverse their course. Those above $s_*$ curve up to
$s_1$ as they approach the critical line, while those below it curve
down. They cannot cross the constant phase line passing through
$s_{**}$, nor can they cross the critical line. They must then head
back to $\sigma=-\infty$, where they will breach the monotonic
nature of the variation of $\arg\Delta_3(2,2 m;s)$ with $t$. Thus,
this situation cannot arise.
\end{proof}

We conclude this section with an investigation of the structure of
lines of constant phase which cut the critical line and turn back to
$\sigma=\infty$ thereafter. We start with the expansion of $\arg
\Delta_3(2,2;\sigma +i t)$ around a point $s_*=\sigma_*+i t_*=1/2+i
t_*$ in the asymptotic region where $\partial \arg
\Delta_3(2,2;\sigma +i t)/\partial \sigma=0$, which is of the form:
\begin{eqnarray}
\arg \Delta_3(2,2;\sigma +i t)&=&\arg \Delta_3(2,2;1/2 +i
t_*)+\frac{2 m^2}{t_*^3}(\sigma-1/2)^2 - \frac{2
m^2}{t_*^3}(t-t_*)^2\nonumber \\& & +\frac{2 m^2}{t_*^2}(t-t_*)
-\left[\frac{\partial^2}{\partial t^2} \log |\Delta_3(2,2;\sigma +i
t)| \right]_{s=s_*} \nonumber \\
& & (\sigma-\sigma_*)(t-t_*)+\ldots . \label{hyp1}
\end{eqnarray}
Here we have employed the asymptotic estimates for $\arg
\Delta_3(2,2;1/2+i t)$ based on (\ref{mz21}).  The corresponding
trajectories of constant phase may be shown to be rectangular
hyperbolae, with  their centre at
\begin{eqnarray}
\sigma-1/2&=&\frac{1}{\frac{1}{4} \left[\frac{\partial^2}{\partial
t^2} \log |\Delta_3(2,2;\sigma +i t)| \right]_{s=s_*}^2+\frac{4
m^4}{t_*^6}} \left(\frac{m^2}{2 t_*^2}\right)
\left[\frac{\partial^2}{\partial t^2} \log |\Delta_3(2,2;\sigma +i
t)| \right]_{s=s_*}  ,\nonumber \\
t-t_*&= &\frac{1}{\frac{1}{4} \left[\frac{\partial^2}{\partial t^2}
\log |\Delta_3(2,2;\sigma +i t)| \right]_{s=s_*}^2+\frac{4
m^4}{t_*^6}} \left(\frac{2 m^4}{ t_*^5}\right).\label{hyp2}
\end{eqnarray}
The second derivative factor in (\ref{hyp2}) is always negative, and
tends to be much larger than the terms involving powers of $1/t_*$.
Thus, we see that the centre of the hyperbolic trajectories of
constant phase will always lie to the left of the critical line,
with its ordinate very close to $t_*$. This displacement of the
centre into $\sigma<1/2$ creates the region in which phase lines can
cut through $\sigma=1/2$ and return to $\sigma>1/2$ via passage
through a zero of $\Delta_3(2,2;s)$. Note that the centre
corresponds to a point at which two lines of constant phase
intersect; thus, it must have derivatives of phase along two
independent lines which are zero. It therefore is a point at which
both $\partial \arg \Delta_3(2,2;\sigma +i t)/\partial t$ and
$\partial \arg \Delta_3(2,2;\sigma +i t)/\partial \sigma$ are zero.
From the Cauchy-Riemann equations, it also is a point at which
$\partial \log |\Delta_3(2,2;\sigma +i t)|/\partial t$ and $\partial
\log |\Delta_3(2,2;\sigma +i t)|/\partial \sigma$ are zero- i.e., it
is a point of extremum for both amplitude and phase.

These hyperbolic centre points are locations at which $\Delta_3
(2,2;s)'=0$. The fact that they lie to the left of $\Re (s)=1/2$ is
is similar to the property proved by Speiser (1934) for $\zeta(s)$.

\begin{figure}
\begin{center}
\includegraphics[width=6cm]{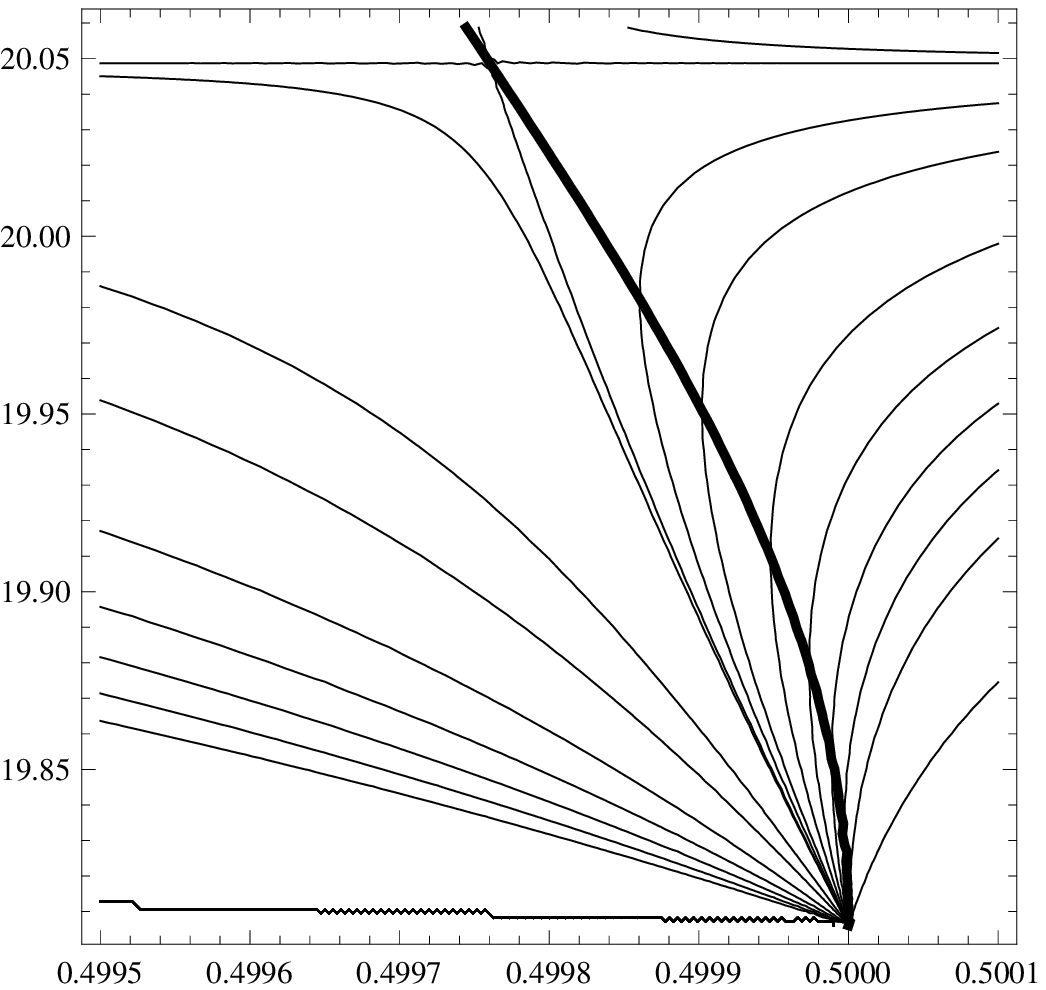}~~\includegraphics[width=6cm]{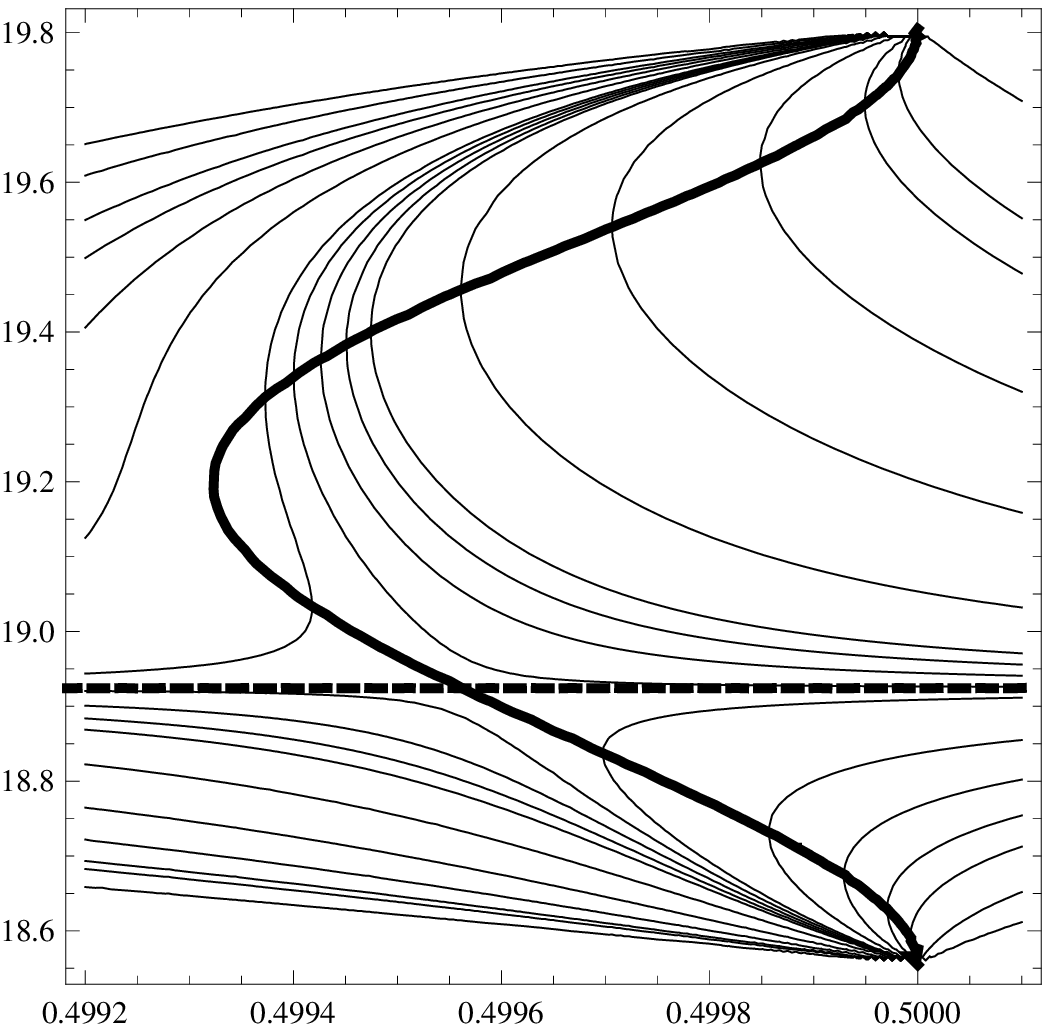}
%~~\includegraphics[width=6cm]{argpltdel314abc.eps}~~
\end{center}
\caption{ Contours of constant phase for $\Delta_3 (2,2;s)$ in the
region around its zero at $s=0.5+i 19.80599$. The thick line is the
contour on which $\partial \arg \Delta_3(2,2;\sigma +i t)/\partial
t$ is zero. (a)The fine lines correspond to the phases 3.04,3.04076,
3.0410, 3.0412, 3.0414, 3.0416, 3.0418, 3.0419, 3.04198, \ 3.042,
3.0425, 3.043, 3.044, 3.045, 3.046, 3.047 and 3.048. (b) The thick
dashed line is the contour on which $\partial \log
|\Delta_3(2,2;\sigma +i t)|/\partial t$ is zero. The contours are
-0.1090, -0.1080, -0.1072, -0.1068, -0.1064, -0.1060, -0.1056,
-0.1055, -0.1054, -0.1053, -0.1052, -0.1048, -0.104, -0.103, -0.102,
-0.1015 and -0.10.} \label{argdeltfine}
\end{figure}

The curves of constant phase of $\Delta_3 (2,2;s)$ given in Fig.
\ref{argdeltfine} illustrate some more of their   general
characteristics in the neighbourhood of zeros.  For Fig.
\ref{argdeltfine} (a), at the top of the figure, we see the centre
of the hyperbolic phase curves; through this centre passes a curve
on which $\partial \arg \Delta_3(2,2;\sigma +i t)/\partial t=0$.
This curve connects the hyperbolic centre to the zero of $\Delta_3
(2,2;s)$ below it, and also continues  to the zero above; at each
zero it is tangent to the $t$ axis. The curve marks points where the
constant phase lines have vertical slope. The lowest line of
constant phase on the right has vertical slope when it arrives at
the zero of $\Delta_3 (2,2;s)$; it corresponds to 3.04076, which is
the value given by (\ref{mz15}) for the phase on the critical line
just above the zero. A second important line has the phase 3.04198,
which is the phase corresponding to the centre of the hyperbola.
This line again connects the hyperbolic centre to the zero. All
lines whose phase lies between these values come in from the right,
cross the zero line of $\partial\arg \Delta_3(2,2;\sigma +i
t)/\partial t$, and curve back to pass through the zero of $\Delta_3
(2,2;s)$. Where they lie to the left of the zero line of
$\partial\arg \Delta_3(2,2;\sigma +i t)/\partial t$, their phase
increases as $t$ decreases; where they lie to the right, it
increases as $t$ increases. Curves coming in from the left all pass
through the zero without crossing the zero derivative line; their
phase always increases as $t$ decreases.

For Fig. \ref{argdeltfine} (b), we show the phase contours below the
zero, in the region down to the next zero. The phase at the centre
of the hyperbolae is -0.1055, while the phase just above the bottom
zero and just below the upper zero are respectively -0.107604 and
-0.101134.  Note the dashed line passing through the centre, along
which  $\partial \log |\Delta_3(2,2;\sigma +i t)|/\partial t=0$.
From the Cauchy-Riemann equations, this is perpendicular to the
solid line defined by  $\partial\arg \Delta_3(2,2;\sigma +i
t)/\partial t=0$ (a fact disguised by the different scales on the
horizontal and vertical axes).

\section{Distributions of zeros}

We return to the left-hand side of expression (\ref{ev6}), in which we replace $s$ by
$\sigma+ \ri t$, and expand assuming $|t|>>\sigma$, with $\sigma$
large enough to ensure accuracy of (\ref{ev6}). The result is
\begin{equation}
2 t\log (t)-2 t(\log
\pi+1)+\pi(\sigma-\frac{1}{2})+\frac{\sigma}{t}(1-2\sigma).
\label{dz1}
\end{equation}
As for each increment of $\pi$ of this expression we get one null
line of the real part of $\Delta_3$ and one of the imaginary part,
and these intersect at $\sigma=1/2$ to give one zero there (assuming
the Riemann hypothesis holds for $\Delta_3(2,2 m;s)$), we can divide
(\ref{dz1}) by $\pi$, and regard the result as a distribution
function for zeros of $\Delta_3$:
\begin{equation}
N_{\Delta_3}(\sigma, t)=\frac{2 t}{\pi} \log (t)-\frac{2
t}{\pi}(1+\log \pi) +\sigma-\frac{1}{2}+\frac{\sigma}{\pi t}
(1-2\sigma). \label{dz2}
\end{equation}

Now, from Titchmarsh and Heath-Brown (1987), the distribution function for the zeros of the  Riemann zeta function on the
critical line is
\begin{equation}
N_{\zeta}(\frac{1}{2},t)=\frac{ t}{2\pi} \log (t)-\frac{
t}{2\pi}(1+\log(2 \pi)) +O(\log t).
\label{dz3}
\end{equation}
We complement this with the numerical estimate from McPhedran {\em et al} (2007) for the distribution
function of the zeros of $L_{-4}(s)$:
\begin{equation}
N_{-4}(\frac{1}{2},t)=\frac{ t}{2\pi} \log (t)-\frac{
t}{2\pi}(1+\log( \pi/2)) +O(\log t).
\label{dz4}
\end{equation}
Adding (\ref{dz3}) and (\ref{dz4}) we obtain the distribution function for the zeros of ${\cal C}(0,1;s)$ (see (\ref{mz13})):
\begin{equation}
N_{{\cal C}0,1}(\frac{1}{2},t)=\frac{ t}{\pi} \log (t)-\frac{
t}{\pi}(1+\log( \pi)) +O(\log t).
\label{dz5}
\end{equation}
When we compare this with (\ref{dz2}),  and use the equation
\begin{equation}
N_{\Delta_3}(\frac{1}{2}, t)=N_{{\cal C}1,4}(\frac{1}{2},t)+N_{{\cal C}0,1}(\frac{1}{2},t),
\label{dz5a}
\end{equation}
it suggests the hypothesis that the distribution function of zeros of ${\cal C}(1, 4;s)$
is the same as that of (\ref{dz5}), to the number of terms quoted:
\begin{equation}
N_{{\cal C}1,4}(\frac{1}{2},t)=N_{{\cal C}0,1}(\frac{1}{2},t)=\frac{ t}{\pi} \log (t)-\frac{
t}{\pi}(1+\log( \pi)) +O(\log t).
\label{dz6}
\end{equation}
Strong numerical evidence supporting this is given in
Table~\ref{table1}, which also shows zero counts for ${\cal
C}(1,8;s)$ and ${\cal C}(1,12;s)$.  Note that the numbers of zeros
found for ${\cal C}(1,4;s)$, ${\cal C}(1,8;s)$ and ${\cal
C}(1,12;s)$ are virtually the same. This rules out any variation
with increasing order similar to that  of Dirichlet $L$ functions,
where increasing order results in significant increases in density
of zeros (compare (\ref{dz4}) and (\ref{dz3}), or the second and
third columns of Table 1).

Comparing the data of Table 1 with the discussion in Bogomolny and
Lebouef (1994), we can see that the split up of $N(t)$ into averaged
parts given by expressions like (\ref{dz3}-\ref{dz6}) and
oscillating parts applies to ${\cal C}(0,1;s)$ and to the ${\cal
C}(1,4 m;s)$. However, it would be value to extend the numerical
investigations of Table 1 to much higher values of $t$, to render
the characterization of the oscillating term more accurate. Such an
extension may require the development of an alternative algorithm to
that based on (\ref{gr1}), which will probably become unwieldy for
values of $t$ of order $10^4-10^5$.

\begin{figure}
\begin{center}
\includegraphics[width=6cm]{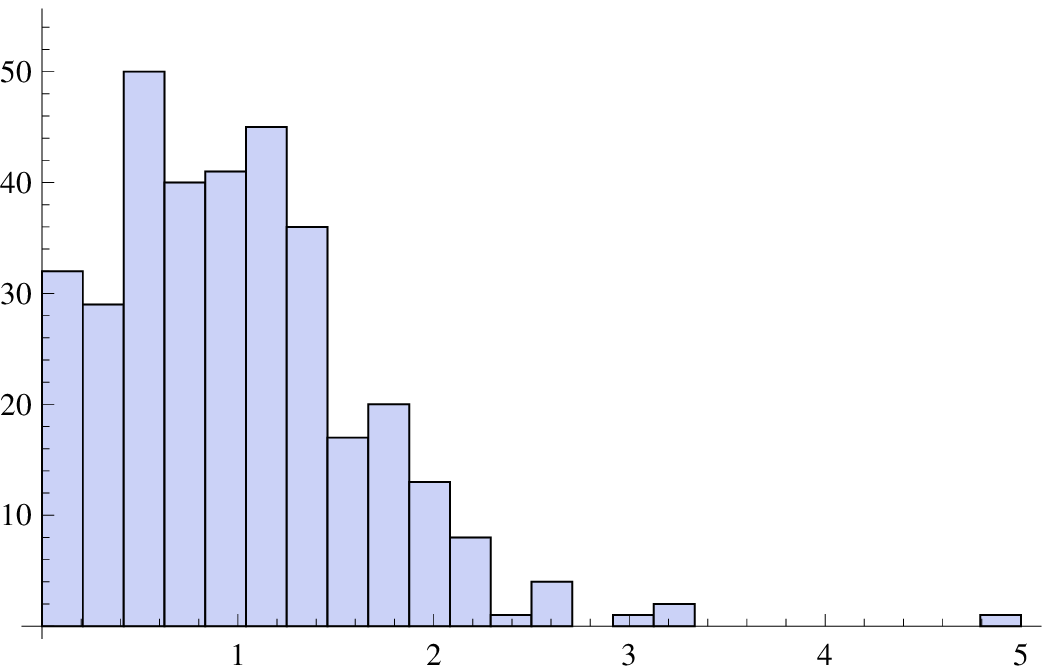}~~\includegraphics[width=6cm]{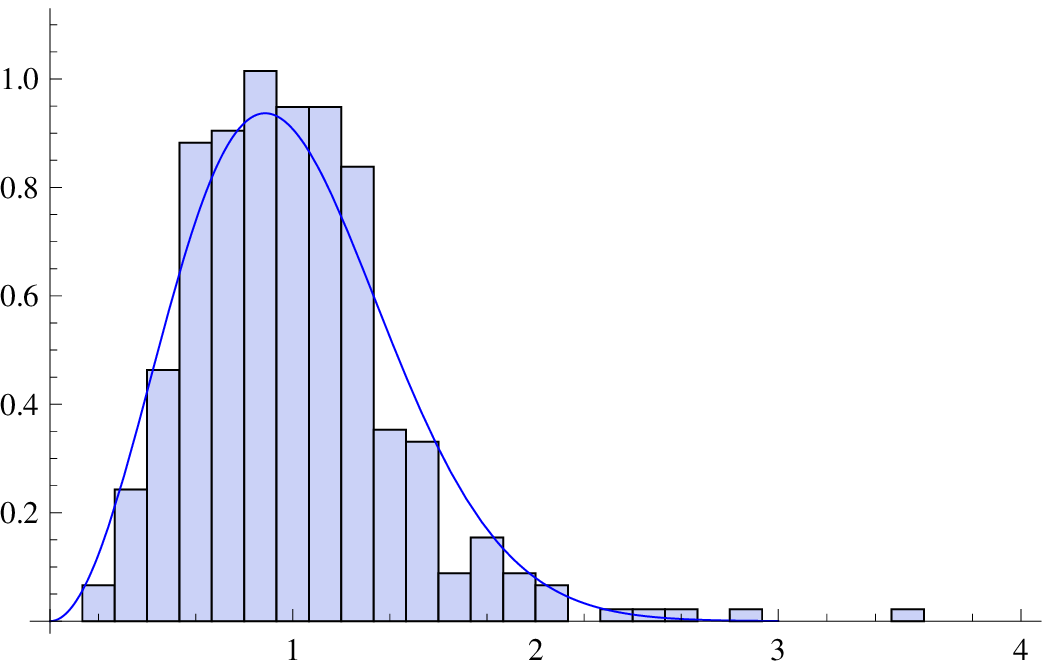}~~
\end{center}
\caption{ The distributions of the differences between successive
zeros of ${\cal C}(0,1;s)$(left) and  ${\cal C}(1,4;s)$ (right) on
$s=1/2+i t$, for $t<300$.} \label{histofig}
\end{figure}

In Fig. \ref{histofig} we compare the distributions of the
differences between zeros on the critical line for ${\cal
C}(0,1;s)$(left) and  ${\cal C}(1,4;s)$ (right). The distributions
are quite different, even with this modest data set. Bogomolny and
Lebouef (1994) have studied the case of  ${\cal C}(0,1;s)$ using
10,000 zeros after $t=10^5$, and contrast the distribution for
$\zeta(s) L_{-4}(s)$ with that for each function separately. The
separate factors in fact have distributions like that that for
${\cal C}(1,4;s)$. The function compared with the histogram in the
right of Fig. \ref{histofig} corresponds to the Wigner surmise,
which (Bogomolny and Lebouef, 1994, Dietz and Zyczkowski, 1991) for
the unitary ensemble takes the form
\begin{equation}
P(S)=\frac{9 S^2}{\pi^2} \exp \left(\frac{-4 S^2}{\pi}\right).
\label{ue}
\end{equation}
Here the separation between zeros has been rescaled to have a mean
of one. Bogomolny and Lebouef (1994) comment that the left
distribution is that of an uncorrelated superposition of two unitary
ensemble sets. Note as one indicator of this that  there is not the
same pronounced tendency for the probability to go to zero with
separation on the left as on the right, where the distribution
clearly comes from a single ensemble. Recall that Table 1 shows that
the frequency distribution for zeros is the same for ${\cal
C}(0,1;s)$ and  ${\cal C}(1,4;s)$. This makes the strong difference
in the distributions of the gaps all the more interesting.

\begin{table}
\caption{Numbers of zeros of $\zeta(1/2+\ri t)$, $L_{-4}(1/2+ \ri
t)$, ${\cal C}(1,4;1/2+\ri t)$,${\cal C}(1,8;1/2+\ri t)$ and ${\cal
C}(1,12;1/2+\ri t)$ in successive intervals of $t$.}
\begin{center}\begin{tabular}{|c|c|c|c|c|c|c|} \\ \hline
$t$ & $n_{\zeta}$ & $n_{-4}$ & $n_{{\cal C}14}$ & $n_{\zeta} +$             &  $n_{{\cal C}18}$ & $n_{{\cal C}112}$ \\
    &             &          &                  & $n_{-4} + n_{{\cal C}14}$ &         &           \\ \hline
0-10 & 0 & 1 & 2 & 3 & 2 & 3\\
10-20 & 1 & 4& 5 & 10  &5 & 5\\
20-30 & 2 & 5 & 6& 13 & 7 & 7 \\
30-40 & 3 & 4& 8 & 15 &8 & 8\\
40-50 & 4 & 6 & 8 & 18 &8 & 8 \\
50-60 & 3 & 5 & 9 & 17& 9 & 9 \\
60-70 & 4 & 6 & 9 & 19 &10 & 10 \\
70-80  & 4 & 6 & 11 & 21 &10 & 10 \\
80-90 & 4 & 7 & 11 & 22&11 & 10\\
90-100 & 4 & 6 & 10 & 20& 10 & 12 \\
\hline 0-100 & 29 & 50 & 79 & 158& 80 & 82\\ \hline
(\ref{dz3}),(\ref{dz4}),(\ref{dz6}) &28 & 50 & 78& 156 &  &
\\\hline
100-110 & 4 & 7 & 11 & 22 & 11 & 10 \\
110-120 & 5 & 7 & 12 & 24 &  12 & 12 \\
120-130 & 5 & 7 & 12 & 24 & 12  & 12 \\
130-140 & 5 & 7 & 12 & 24 &  12 & 12 \\
140-150 & 4 & 7 & 12 & 23 &  12  & 11 \\
150-160 & 6 & 7 &  12  &  25  & 12 & 13    \\
160-170 & 6 & 7 &  13  &  26  & 13 & 12     \\
170-180 & 6 & 8 & 13    &  27  & 13  & 13  \\
180-190 & 5 & 8 & 13   & 26   & 12  & 14  \\
190-200 & 5 & 7 & 13   &  25  &  14  & 13 \\
\hline 0-200 & 80 & 122 & 202 & 404 & 203  & 204  \\ \hline
(\ref{dz3}),(\ref{dz4}),(\ref{dz6}) &78 & 122 & 200& 400  & &
\\\hline
200-210 & 6 & 8 & 13 & 27 & 12 & 13 \\
210-220 & 5 & 8 & 14 & 27 &  15 & 14 \\
220-230 & 6 & 8 & 13 & 27 & 13  & 13 \\
230-240 & 6 & 8 & 14 & 28 & 14  & 14  \\
240-250 & 6 & 8 & 14 & 28 & 14  & 13 \\
250-260 & 6 & 8 & 13   & 27   & 14 & 15    \\
260-270 & 6 & 8 & 15   & 29   & 14 & 13     \\
270-280 & 6  & 8 & 14    &  28  & 14  & 15  \\
280-290 & 6  & 8 &15    & 29   & 15  & 14  \\
290-300 & 6 & 9 &  14  &  29  &  13  & 15 \\
\hline 0-300 & 137 & 203 & 341 &681  &341   &  343 \\ \hline
(\ref{dz3}),(\ref{dz4}),(\ref{dz6}) &137 & 203 & 340&680  & &
\\\hline
\end{tabular}
\end{center}
\label{table1}
\end{table}
\section{Relation between Distributions of Zeros}
We now investigate the relationship between the distributions of zeros of ${\cal C}(1,4 m;s)$ and ${\cal C}(0,1;s)$. To do this,
we consider, rather than the product of these functions, as in $\Delta_3 (2,2 m;s)$, their quotient
\begin{equation}
\Delta_4(2,2 m;s)=\frac{{\cal C}(1,4 m;s)}{{\cal C}(0,1;s)}.
\label{eq71}
\end{equation}
We will now prove some of the important properties of this function.

\begin{theorem}
The analytic function $\Delta_4(2,2 m;s)$ obeys the functional equation
\begin{equation}
\Delta_4(2,2 m;s)=\Delta_4(2,2 m;1-s) {\cal F}_{2 m}(s).
\label{eq72}
\end{equation}
It has first order poles at $s=-(2 m-1),-(2 m-2),\ldots,-1$ and a first order zero at $s=0$. Its only possible essential singularity is at infinity.
On the critical line, it lies in either the first or third quadrants, with its argument being given by
\begin{equation}
\arg [\Delta_4(2,2 m;1/2+i t)]=\phi_{2 m,c}(t)-\left[ \begin{tabular}{c}
0\\
$\pi$ \\
\end{tabular}\right].
\label{eq73}
\end{equation}
As $\sigma\rightarrow \infty$ for any  $t$, the argument of $\Delta_4$ tends to zero exponentially, while as $\sigma\rightarrow -\infty$
for any  $t$,the argument of $\Delta_4$ tends to zero algebraically.
\end{theorem}
\begin{proof}
The functional equation (\ref{eq72}) follows readily from the functional equations (\ref{mz4}) for ${\cal C}(1,4 m;s)$ and
${\cal C}(0,1;s)$. The poles on the negative real axis arise from the poles of ${\cal F}_{2 m}(s)$, while
the zero at $s=1$ arises from the first order pole there of ${\cal C}(0,1;s)$. The function ${\cal C}(0,1;s)$ has an infinite set of zeros on the critical line which it inherits from $\zeta (s)$ (Titchmarsh \& Heath-Brown, 1987). This will make the point at infinity an essential singularity unless all but a finite number of these coincide with zeros of ${\cal C}(1, 4 m;s)$.
The phase of $\Delta_4(2,2 m;s)$ on the critical line follows from (\ref{eq72}) when $1-s=\overline{s}$.

For $\sigma$ positive and not small, we may expand ${\cal C}(1,4 m;s)$ and
${\cal C}(0,1;s)$ by direct summation:
\begin{equation}
{\cal C}(0,1;s)=4 \left( 1+\frac{1}{2^s}+\frac{1}{4^s}+2\frac{1}{5^s}+\ldots \right),
\label{eq74}
\end{equation}
and
\begin{equation}
{\cal C}(1,4 m;s)=4 \left( 1+(-1)^m \frac{1}{2^s}+\frac{1}{4^s}+2\frac{[\cos (4 m \cos^{-1}(1/\sqrt{5}))+ \cos (4 m \cos^{-1}(2/\sqrt{5}))]}{5^s}+\ldots \right).
\label{eq75}
\end{equation}
Combining (\ref{eq74}) and (\ref{eq75}), we find for $m=1,2,\ldots$
\begin{equation}
\Delta_4(2,4 m-2;s)=1-\frac{1}{2^{s-1}}+O(\frac{1}{4^s}),
\label{eq76}
\end{equation}
and
\begin{equation}
\Delta_4(2,4 m;s)=1+\frac{[\cos (8 m \cos^{-1}(1/\sqrt{5}))+ \cos (8 m \cos^{-1}(2/\sqrt{5}))-2]}{5^s}+O(\frac{1}{8^s}).
\label{eq77}
\end{equation}
Thus, in either case, $\Delta_4$ tends to unity in exponential fashion as $\sigma$ increases positively. The difference from unity in the case
of order $4 m-2$ tends to zero as $\exp(-\sigma \log(2))$, and in the case of order $4 m$ as $\exp(-\sigma \log(5))$, so that the phase of
$\Delta_4$ tends to zero exponentially.

For $\sigma$ negative and not small, we may take the phase of $\Delta_4(2,4 m;1-s)$ to be zero, and write
\begin{equation}
\arg [\Delta_4(2,4 m;s)]\simeq \arg [{\cal F}_{2 m}(s)]=2\Re [\phi_{2 m}(s)]=\frac{4 m^2 t}{(\sigma-1/2)^2+t^2}+O(1/t^3).
\label{eq78}
\end{equation}
This quantity tends to zero algebraically as $|s|$ increases. Its maximum value is $4 m^2/(1/2-\sigma)$, which occurs when $t=1/2-\sigma$.
\end{proof}

If necessary, it is not hard to construct uniform bounds on the remainder terms in series like that of (\ref{eq75}) for ${\cal C}(1,4 m;s)$ in the region of absolute convergence $\sigma >1$. Indeed, suppose we break the series there into a part ${\cal F}_N(1, 4 m;s)$ with $\sqrt{p_1^2+p_2^2}\leq N$, and a remainder term  ${\cal R}_N(1, 4 m;s)$ from  $\sqrt{p_1^2+p_2^2}> N$:
\begin{equation}
|{\cal C}(1,4 m;s)|\leq |{\cal F}_N(1, 4 m;s)|+|{\cal R}_N(1, 4 m;s)|,
\label{eest1}
\end{equation}
and $|\cos( 4 m \theta_{p_1,p_2})|\leq 1$ for every $m$, $p_1$, $p_2$, so that
\begin{equation}
|{\cal F}_N(1, 4 m;s)|\leq |{\cal F}_N(1, 0;s)| , ~ |{\cal R}_N(1, 4 m;s)|\leq |{\cal R}_N(1, 0;s)|.
\label{eest2}
\end{equation}
In turn,
\begin{equation}
|{\cal F}_N(1, 0;s)|\leq |{\cal F}_N(1, 0;\sigma)| , ~ |{\cal R}_N(1, 0;s)|\leq |{\cal R}_N(1, 0;\sigma)|,
\label{eest3}
\end{equation}
so bounding uniformly the truncated series, the remainders and the ${\cal C}(1,4 m;s)$.
\begin{theorem}
The only lines of constant phase of the $\Delta_4(2,2 m;s)$ which can attain $\sigma=\infty$ are equally spaced, and have interspersed lines
of constant modulus.
All such lines of constant phase either reach the critical line in the asymptotic region of $t$ at a pole or a zero of $\Delta_4(2,2 m;s)$, or curve up and asymptote to each other at infinity.
\end{theorem}
\begin{proof}
From (\ref{eq76}), we have
\begin{equation}
\Delta_4(2,4 m-2;\sigma+i t)=1-\frac{e^{-i t \log 2}}{2^{\sigma-1}}+O(\frac{1}{4^{\sigma+i t}}),
\label{eq79}
\end{equation}
and so the leading order estimate gives the lines of phase zero for $\Delta_4(2,4 m-2;\sigma+i t)$ as occurring at $t=n\pi/\log 2$, for
$n=0,1,2,\ldots$. Halfway between these lines of phase zero are the lines on which the leading order estimate gives
$\partial\Delta_4(2,4 m-2;\sigma+i t)/\partial t=0$: these are lines of constant modulus, and in fact correspond to $|\Delta_4(2,4 m-2;\sigma+i t)|=1$,
for $\sigma\rightarrow \infty$. The same argument applies to $\Delta_4(2,4 m;\sigma+i t)$, with $\log 5$ replacing $\log 2$ in the estimate for
asymptotic placement of lines of phase zero and amplitude unity.

Next, consider where these lines of constant phase and amplitude can go as $\sigma$ decreases towards $1/2$. Along the line of unit amplitude, as $\sigma$ decreases curves of constant phase emanate from the line at right angles and head above and below it towards smaller values of $\sigma$. We see that the phase of $\Delta_4(2,2 m;s)$ must vary monotonically along the line of unit amplitude; were it to be otherwise, lines of constant phase would form closed loops about the line, constraining the phase to be everywhere constant within the closed loop- a contradiction. On either side of each line of phase zero, we have one line of unit amplitude whose phase increases monotonically as $\sigma$ decreases, and another line of unit amplitude along which the phase decreases monotonically.  A similar argument applies to the lines of phase zero, along which the amplitude of $\Delta_4(2,2 m;s)$ either increases monotonically or decreases monotonically as $\sigma$ decreases. The phase of $\Delta_4(2,2 m;s)$ on the lines of unit amplitude moves away monotonically from zero, so that these lines can never intersect lines of phase zero in the finite part of the plane (although they could asymptote towards such lines at $s=\infty$). They cannot return to $\sigma=\infty$, since their phase cannot tend back towards zero. Thus, the lines of constant amplitude and phase coming from $\sigma=\infty$ must all cut the critical line, or they could from a certain point on all commence to curve up towards $t=\infty$.

Those lines of constant phase zero which reach the critical line in the region where the phase of ${\cal F}_{2 m}(s)$ is accurately constrained by the asymptotic estimate (\ref{mz21}) can only cut the line at a pole or zero of $\Delta_4(2,2 m;s)$.
\end{proof}

We now give one of the two principal results of this section.
\begin{theorem}
Suppose ${\cal C}(0,1;s)$ obeys the Riemann hypothesis. Then ${\cal C}(1, 4 m;s)$ obeys the Riemann hypothesis for any positive integer $m$. Conversely, if ${\cal C}(1, 4 m;s)$ obeys the Riemann hypothesis for some $m$, then ${\cal C}(0,1;)$ obeys the Riemann hypothesis.
\end{theorem}
\begin{proof}
Consider lines $L_1$ and $L_2$  in $t>0$ along which the phase of $\Delta_4(2,2 m;s)$ for a given $m$ is zero. Join these
lines with two lines  to the right of the critical line along which  $\sigma$ is constant. Then the change of argument of  $\Delta_4(2,2 m;s)$ around the closed contour  $C$ so formed is zero, so by the Argument Principle the number of poles inside the contour equals the number of zeros. Each pole is formed by a zero of ${\cal C} (0,1;s)$, so if there are no such zeros within $C$ there can be no zeros of  ${cal C}(1, 4 m;s)$ within $C$. We can repeat this procedure for all $m$.

Conversely,  each zero of $\Delta_4(2,2 m;s)$ is formed by a zero of  ${cal C}(1, 4 m;s)$, so if
there are no such zeros for some $m$, then there can be no poles, and hence no zeros of ${\cal C}(0,1;s)$. These arguments prove the theorem in the region to the right of the critical line lying between lines of zero phase of $\Delta_4(2,2 m;s)$, with the result to the left of the critical line then guaranteed by the functional equation (\ref{eq72}) .

To complete the proof we need to show that any point in the region $\sigma>1/2$, $t>0$ is enclosed between lines of phase zero of  $\Delta_4(2,2 m;s)$ coming from $\sigma=\infty$. We note that $t=0$ is one such line, and that for any $\sigma>0$ the infinite number of such constant phase lines cannot cluster into a finite interval of $t$, since that would indicate an essential singularity of  $\Delta_4(2,2 m;s)$ for that $\sigma$.
\end{proof}

\begin{corollary}
Assuming the Riemann hypothesis applies to ${\cal C}(1,4 m;s)$
or ${\cal C}(0,1;s)$, apart from a small number of exceptions, lines of  phase equal to $0$ or $\pi$ of  $\Delta_4(2,2 m;s)$ which
leave the critical line going into $\sigma<1/2$  cannot reach the asymptotic
region in which the phase estimate (\ref{eq78}) is accurate. Instead, they must loop back to cut the critical line.
\end{corollary}
\begin{proof}
If $t$ is sufficiently large to ensure $\frac{4 m^2 t}{(\sigma-1/2)^2+t^2}<\pi$, then no line with phase $\pi$ can enter the asymptotic region.
Lines with phase zero are also excluded.
Hence, lines leaving $\sigma=1/2$ with $t$ large and going left which do not return to the critical line must either curve upwards and go to
$t=\infty$ or curve downwards and cut the line $t=0$. If the first alternative applies, then the image of this line under $s\rightarrow 1-s$ in the functional
equation is a line of variable phase running upwards in $\sigma>1/2$, which must cut lines of phase zero running left to intersect the critical line.
Such intersection points would have to be zeros or poles of $\Delta_4(2,2 m;s)$, contradicting our assumption. The second alternative means
the line of constant phase must intersect the axis $t=0$ at one of the $2 m-1$ poles on the axis. Thus, the exceptional cases are limited to
the region $t<4 m^2$ and to lines passing through the $2 m-1$ poles.

In fact, it may be numerically verified in particular cases whether such exceptions do in fact occur. They do not for the three cases we
have investigated ($m=1,2,3$).
\end{proof}

We are now in a position to prove the second of the main results of this section.

\begin{theorem}
Assuming the Riemann hypothesis applies to ${\cal C}(1,4 m;s)$
or ${\cal C}(0,1;s)$, then given any two lines of phase zero of $\Delta_4(2,2 m;s)$ running from $\sigma=\infty$
and intersecting the critical line, the number of zeros
and poles of $\Delta_4(2,2 m;s)$ counted according to multiplicity and lying properly between the lines must be the same.
\end{theorem}
\begin{proof}
We consider a contour composed of the two lines of phase zero, the segment between them on the critical line
and a segment between them in the region $\sigma>>1$. The total phase change around this contour is strictly zero, since the region
$\sigma>>1$ has the phase of $\Delta_4(2,2 m;s)$ constrained: $-\pi <<arg[\Delta_4(2,2 m;s)]<<\pi$. More particularly, if $P_u=(1/2,t_u)$
lies at the upper end of the segment on the critical line, and $P_l=(1/2,t_l)$ at the lower end, the total phase change between a point
approaching $P_u$ on the contour from the right and a point leaving $P_l$ going right is zero. This phase change is made up of contributions
from the changes of phase at the zero or pole $P_u$, from the zero or pole $P_l$, from  the $N_z$ zeros  and $N_p$ poles
on the critical line between $P_u$ and $P_l$, and from the phase change between the zeros and poles. In this list, the first change is
$\phi_{2m,c}(t_u)$, the phase on the critical line just below $P_u$. (We could also have a phase $\phi_{2m,c}(t_u)-\pi$, but it will be easily seen
that in this alternative case the argument which follows will arrive at exactly the same conclusion.) The second change is $-\phi_{2m,c}(t_l)$,
where $\phi_{2m,c}(t_l)$ is
the phase just above $t_l$. Giving  zero $n$ a multiplicity $z_n$, and  pole $n$ a multiplicity $p_n$, the phase change at the former
is $-\pi z_n$ and the latter $\pi p_n$. The phase change between zeros and poles is $\phi_{2m,c}(t_l)-\phi_{2m,c}(t_u)$. Hence,
the phase constraint is
\begin{equation}
\phi_{2m,c}(t_u)-\phi_{2m,c}(t_l)-\pi[\sum_{n=1}^{N_z} z_n -\sum_{n=1}^{N_p} p_n]+\phi_{2m,c}(t_l)-\phi_{2m,c}(t_u)=0,
\label{eq710}
\end{equation}
leading to
\begin{equation}
\sum_{n=1}^{N_z} z_n =\sum_{n=1}^{N_p} p_n,
\label{eq711}
\end{equation}
as asserted.
\end{proof}

%\begin{equation}
%N_{{\cal C}1,4}(\frac{1}{2},t)=N_{{\cal C}0,1}(\frac{1}{2},t)=\frac{ t}{\pi} \log (t)-\frac{
%t}{\pi}(1+\log( \pi)) +O(\log t).
%\label{dz6}
%\end{equation}

\begin{corollary}
If all zeros and poles on the critical line have multiplicity one, the numbers of zeros and poles on
the critical line between any pair of lines of phase zero of $\Delta_4(2,2 m;s)$ coming from $\sigma=\infty$ are the same.
 The distribution functions for zeros $N_{{\cal C}1,4}(\frac{1}{2},t)$ and $N_{{\cal C}0,1}(\frac{1}{2},t)$ of (\ref{dz6})
 then must agree in all terms which go to infinity with $t$.
\end{corollary}
\begin{proof}
The first assertion is a simple consequence of Theorem 7.5. The second follows from Theorem 7.2 and Theorem 7.5: the number of zeros and
poles between successive zero lines coming from $\sigma=\infty$ match for all such pairs of lines, and there are only a finite number of
exceptional poles and zeros which may disturb the equality between numbers of zeros and poles.
If lines of phase zero curve up and asymptote to the point at infinity, we apply Theorem 7.5 to the
leftmost of these and the last line of phase zero cutting the critical line.
\end{proof}

We now show in Fig. \ref{del4fig1} and \ref{del4fig2} some examples of morphologies of lines of constant phase
of $\Delta_4(2,2;s)$ near the critical line. Fig. \ref{del4fig1} gives at top left a diagram in which there are three lines of phase zero coming
from $\sigma=\infty$. The first two give a loop going right which intersects the critical line at a pole, below, and a zero above. The third
again intersects the critical line at a pole, and ends what we will describe as a "cell". By this we mean what is a repeat unit in the loose sense;
above the third line we start another loop going right.
Fig. \ref{del4fig1} shows a view at a structure for higher $t$  in the upper right graph, in which the resolution of contour lines
is insufficient to give a correct impression of the relationship between the lines of constant phase. The magnified views below show at left that the
apparent intersection of lines of phase 0 and $\pi$ is in fact an avoided crossing, where the lines behave in hyperbolic fashion. In the
graph below right, we see a small phase loop going left away from $\sigma=1/2$, with an even smaller line element of phase 0 occurring in $\sigma>1/2$.
Note that, denoting poles and zeros by prefixes $P$ and $Z$, the values of $t$ at which zeros and poles occur in the upper graph are:\\
$P45.6, Z45.9, Z46.9, Z47.71, P47.74, P48.0, Z49.2, P49.72, P49.77$.\\
This data then shows a succession of three zeros of ${\cal C}(0,1;s)$ uninterrupted by zeros of ${\cal C}(1,4;s)$. We have encountered cases where the
opposite situation occurs, but no sequences  of four successive zeros or poles.

\begin{figure}
\begin{center}
\includegraphics[width=6cm]{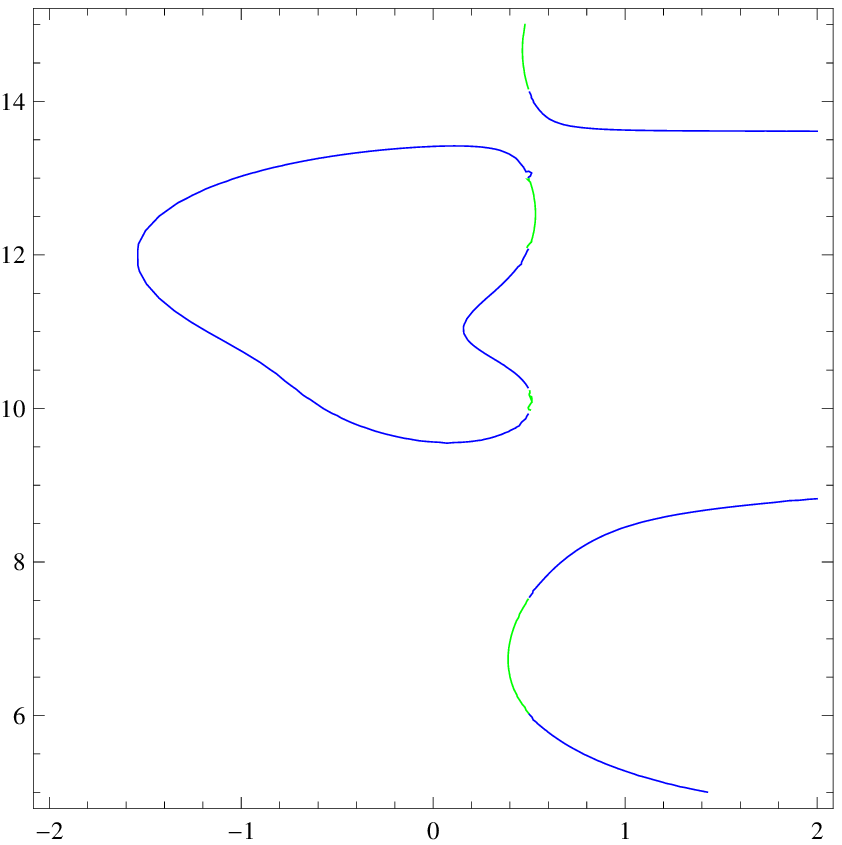}~~\includegraphics[width=6cm]{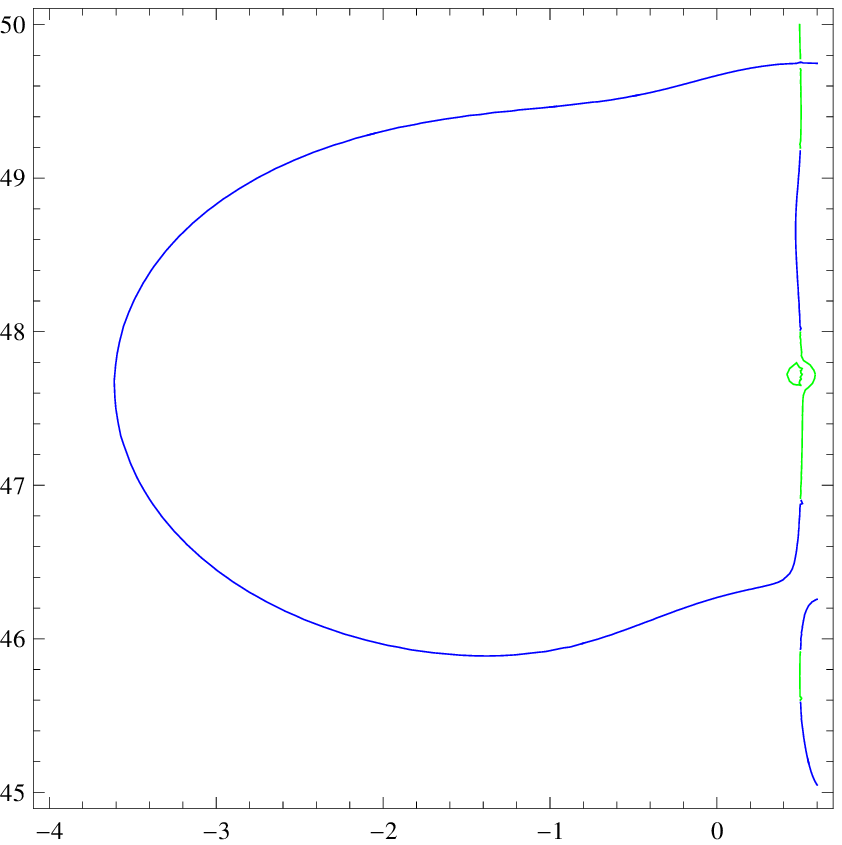}\\
\includegraphics[width=6cm]{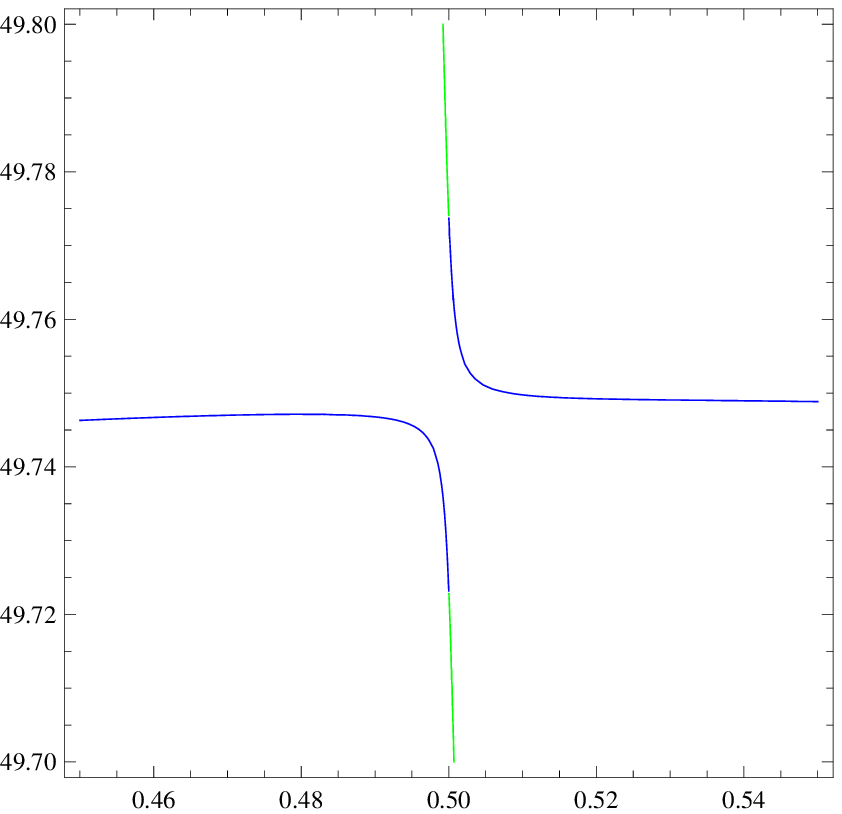}~~\includegraphics[width=6cm]{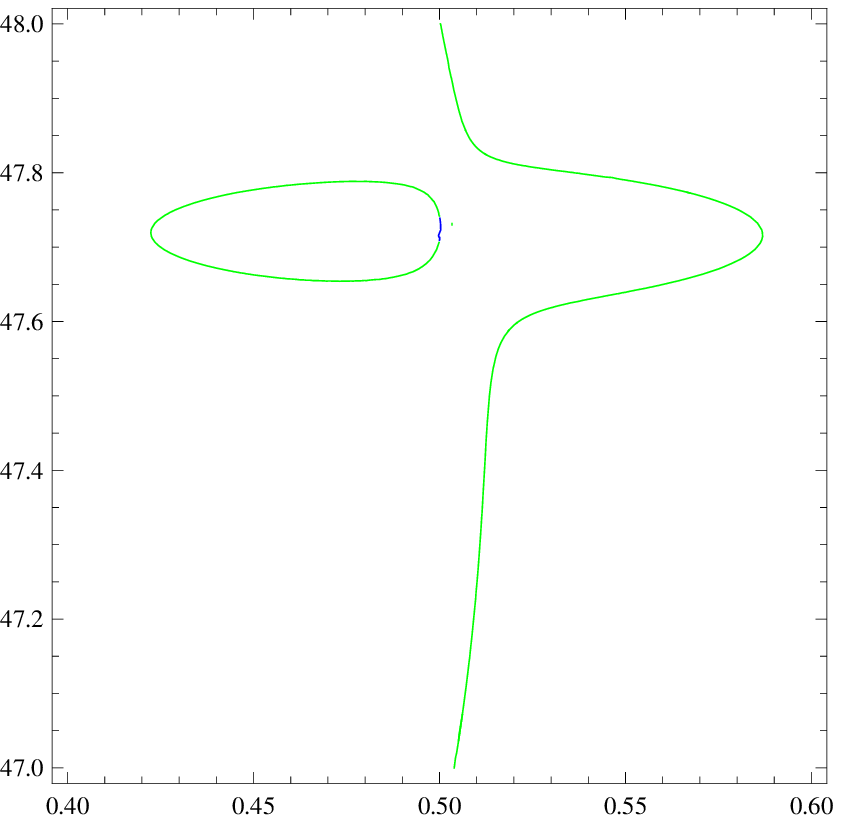}
\end{center}
\caption{ Contours of constant phase for $\Delta_4(2,2;s)$, with phase zero in blue, and phase $\pm \pi$ in green.} \label{del4fig1}
\end{figure}

The leftmost graph in Fig. \ref{del4fig2} shows the hyperbolic distribution of lines of constant phase occurring around a zero
of $\Delta '(2,2;s)$. The discussion following Theorem 5.6 may be extended from $\Delta_3(2,2 m;s)$ to $\Delta_4(2,2 m;s)$, with the result
that between successive zeros of $\Delta_4(2,2 m;s)$ uninterrupted by a pole there will be a hyperbolic point to the right of the critical line.
Between successive poles of $\Delta_4(2,2 m;s)$ uninterrupted by a zero there will be a hyperbolic point to the left of the critical line.
The leftmost graph shows the case of two poles, with a hyperbolic point in $\sigma<0.5$. The rightmost graph in Fig. \ref{del4fig2}
shows the morphology of lines of constant phase for $\Delta_4(2,4;s)$, for which the density of lines of phase zero coming from
$\sigma=\infty$ is, as expected, greater than for  $\Delta_4(2,2 ;s)$.
%a case where the left-going and right-going  loops are almost symmetric under a reflection in the critical line, except of course that the
%right-going loop contains a pole at $t=92.49$ and a zero at $t=92.95$, so that between these the loop penetrates $\sigma<1/2$.

\begin{figure}
\begin{center}
\includegraphics[width=6cm]{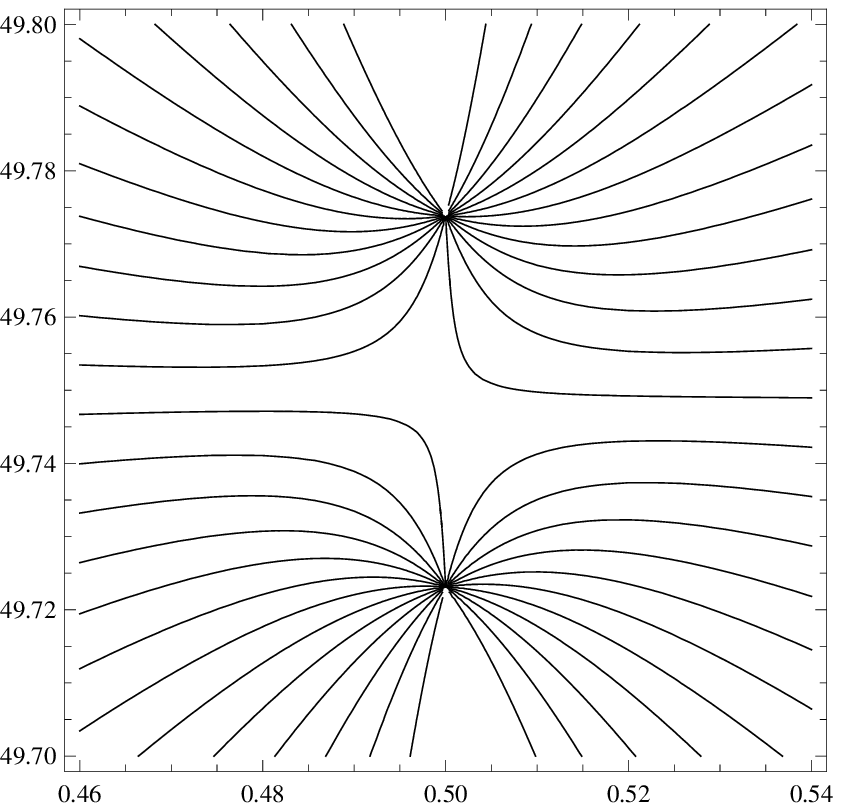}~~\includegraphics[width=6cm]{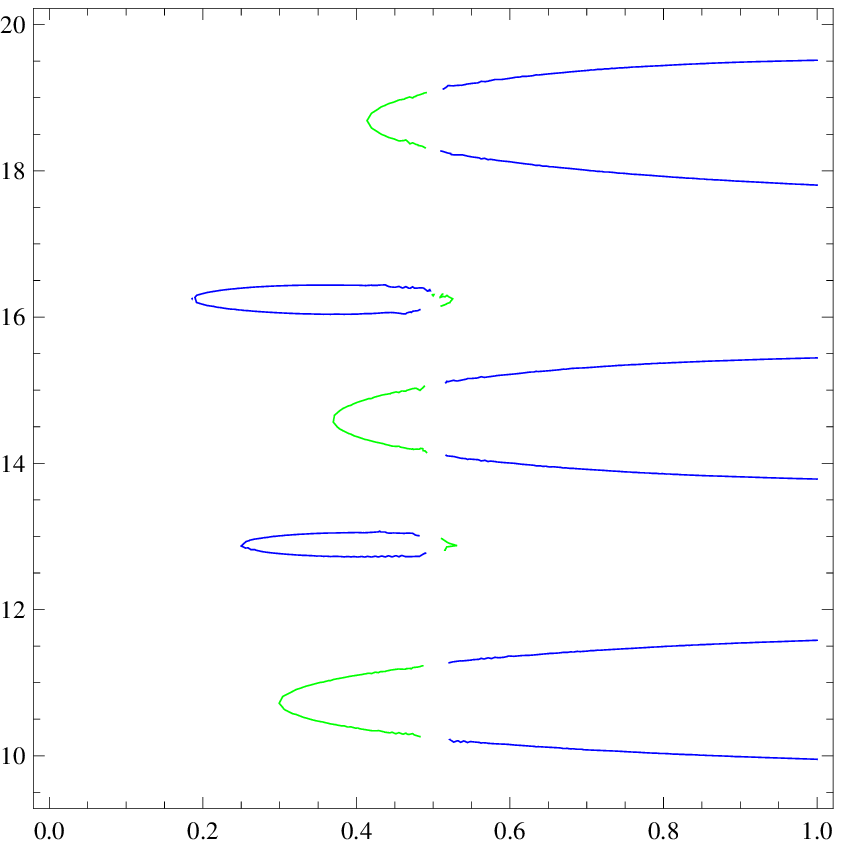}~~
\end{center}
\caption{ Contours of constant phase for $\Delta_4(2,2;s)$,(left) near a hyperbolic point,
(right) with phase zero in blue, and phase $\pm \pi$ in green.} \label{del4fig2}
\end{figure}

Table 2 gives the numbers of zeros and poles in successive cells, for the first twenty-two cells of $\Delta_4(2,2;s)$. These evidently
satisfy the requirements of Theorem 7.4. Note the tendency of the
numbers of zeros and poles to increase with cell number; this is a consequence of the cell length tending to be roughly constant, while the
density of poles and zeros increases logarithmically, in keeping with equation (\ref{dz6}). For the fourth column, the mean length is 9.01
and the standard deviation is 0.90. Note that the corresponding separation at $\sigma=\infty$ is $2\pi/\log (2)$, or around 9.06, so that there
is no significant change in average spacing with variation of $\sigma$. For $\Delta_4(2,4;s)$, Table 3 gives data for the first twenty-four cells,
with the mean cell length being 3.92 and its standard deviation 0.53. The separation at infinity here is $2\pi/\log (5)$, or around 3.90.

\begin{table}
%\begin{minipage}[b]{0.5\linewidth}
\caption{Numbers of zeros and poles,  interval in $t$, and starting point on the critical line, for  the first twenty-two cells lying between lines of phase
zero coming from $\sigma=\infty$, for $\Delta_4(2,2;s)$.}
\centering
\begin{tabular}{| c| c| c| c| c|}
\hline Cell\# & \#Zeros & \#Poles & Length & Lowest pole \\
[0.5ex] \hline
1 & 3 & 3 & 8.11 & 6.02 \\
2 & 5 & 5 & 9.15 & 14.13 \\
3 & 7 & 7 & 9.66 & 23.28 \\
4 & 6 & 6 & 7.98 & 32.94 \\
5 & 8 & 8 & 8.85 & 40.92 \\
6 & 9 & 9 & 10.65 & 49.77 \\
7 & 8 & 8 & 7.95 & 60.42 \\
8 & 9 & 9 & 8.77 & 68.37 \\
9 & 10 & 10 & 10.29 & 77.14 \\
10 & 10 & 10 & 8.71 & 87.43 \\
11 & 9 & 9 & 8.19 & 96.14 \\
12 & 11 & 11 & 9.99 & 104.33 \\
13 & 10 & 10 & 8.63 & 114.32 \\
14 & 10 & 10 & 8.14 & 122.95 \\
15 & 12 & 12 & 10.16 & 131.09 \\
16 & 11 & 11 & 9.05 & 141.25 \\
17 & 10 & 10 & 8.41 & 150.30 \\
18 & 11 & 11 & 8.47 & 158.71 \\
19 & 13 & 13 & 10.52 & 167.18 \\
20 & 11 & 11 & 7.90 & 177.70 \\
21 & 11 & 11 & 8.63 & 185.60 \\
22 & 13 & 13 & 9.99 & 194.23 \\ [1ex] \hline
\end{tabular}
\label{table:Del414}
%\end{minipage}
%\hspace{2.6cm}
%\begin{minipage}[b]{0.5\linewidth}
\end{table}
\begin{table}
\caption{Numbers of zeros and poles,  interval in $t$, and starting point on the critical line, for  the first twenty-four cells lying between lines of phase
zero coming from $\sigma=\infty$, for $\Delta_4(2,4;s)$.}
\centering
\begin{tabular}{| c| c| c| c| c|}
\hline Cell\# & \#Zeros & \#Poles & Length & Lowest pole \\
[0.5ex] \hline
1 & 1 & 1 & 4.22 & 6.02 \\
2 & 2 & 2 & 3.89 & 10.24 \\
3 & 2 & 2 & 4.16 & 14.13 \\
4 & 2 & 2 & 3.16 & 18.29 \\
5 & 3 & 3 & 4.28 & 21.45 \\
6 & 2 & 2 & 3.93 & 25.73\\
7 & 3 & 3 & 3.28 & 29.66 \\
8 & 3 & 3 & 4.65 & 32.94 \\
9 & 3 & 3 & 3.33 & 37.59 \\
10 & 4 & 4 & 4.68 & 40.92 \\
11 & 4 & 4 & 4.17 & 45.60 \\
12 & 3 & 3 & 3.20 & 49.77 \\
13 & 3 & 3 & 3.96 & 52.97 \\
14 & 4 & 4 & 3.90 & 56.93 \\
15 & 4 & 4 & 4.28 & 60.83 \\
16 & 3 & 3 & 3.26 & 65.11 \\
17 & 4 & 4 & 3.79 & 68.37 \\
18 & 4 & 4 & 4.54 & 72.16 \\
19 & 4 & 4 & 3.51 & 76.70 \\
20 & 5 & 5 & 4.53 & 80.21 \\
21 & 3 & 3 & 2.89 & 84.74 \\
22 & 4 & 4 & 4.61 & 87.63 \\
23 & 5 & 5 & 3.90 & 92.24 \\
24 & 4 & 4 & 4.00 & 96.14 \\ [1ex] \hline
\end{tabular}
\label{table:Del418}
%\end{minipage}
\end{table}

\begin{corollary}
If all zeros and poles on the critical line have multiplicity
one, on any closed contour (including those closed at $\sigma
\rightarrow \infty)$ there must be an equal number of zeros and poles
and they must alternate. i.e. following the contour in a clockwise
direction, any pole will be followed by a zero and zero will be
followed by a pole. Every cell begins with a pole on the critical line.
\end{corollary}
\begin{proof}
We follow the line of zero phase coming in from $\sigma=\infty$ which forms the start of the cell. It must intersect the
critical line, in accordance with Corollary 7.2 and meet there a line of
constant $\pi$ phase going towards increasing $t$, which is also a branch cut discontinuity. The
phase $\phi_{2 m,c}(t)$ is decreasing as $t$ increases along the critical line
and above the intersection point (call it $t_1$) there is a change of $\pi$ to the
phase on the critical line, forcing the phase to drop to
$\phi_{2m,c}-\pi$, which  lies between $(-\pi,0)$. As phase is increasing in a
clockwise direction around the point of intersection, that point must be a
pole (i.e., every cell begins with a pole).

We continue to follow the branch cut as it loops back to the critical line
in accordance with Corollary 7.4. The phase below the next intersection
point $t_2$ is still negative as $\phi_{2 m,c}(t)$ decreases along the critical line.
The change of $\pi$ to the phase at $t_2$ must be a positive change, such that
the phase above the point of
intersection must  be $\phi_{2 m,c}(t_2)>0$ . The phase is increasing in an
anti clockwise direction, which implies that the intersection point $t_2$
is a zero. This argument is continued for successive poles and zeros until we reach the end of the cell in question.
\end{proof}

{\bf Remarks:}\\
The following consequences are clear from the arguments in Corollary
7.7:

\begin{itemize}
\item Any closed contour intersecting the critical line at a finite set of points will have  maximum and minimum intersections
of opposite type i.e a zero and a pole.
\item If there are two consecutive zeros or poles on
the critical line they cannot lie on the same closed contour of zero
phase.
\item Any closed loop of zero phase going from $\sigma=1/2$ into $\sigma<\frac{1}{2}$ must have a zero at its lowest intersection
with $\sigma=1/2$.
\end{itemize}

\begin{acknowledgements}
The research of R.McP.  on this project was supported by the
Australian Research Council Discovery Grants Scheme.
\end{acknowledgements}

\newpage
\appendix{Supplementary Notes on Systems of Angular Sums}

We give here results linking trigonometric lattice sums of order up
to ten, which show that they may be generated from three independent
sums, ${\cal C}(0,1;s)$, ${\cal C}(1,4;s)$, and ${\cal C}(1,8;s)$.
We also give expressions in terms of sums of Macdonald functions of
the Kober-type from which these independent sums may be calculated.

A basic result we use relies on the symmetry of the square lattice:
\begin{equation}
{\cal C}(1, 4 m-2;s)=\sum_{p_1,p_2}{'}\frac{\cos (4 m-2)
\theta_{p_1,p_2}}{(p_1^2+p_2^2)^s}=0={\cal C}(2,2 m-1;s)-{\cal
S}(2,2 m-1;s), \label{sn1}
\end{equation}
from which we obtain
\begin{equation}
{\cal C}(2,2 m-1;s)={\cal S}(2,2 m-1;s)=\frac{1}{2}{\cal C}(0,1;s),
\label{sn2}
\end{equation}
for $m$ any positive integer. We can also expand $\cos (4 m-2)
\theta_{p_1,p_2}=T_{4 m-2}(\cos \theta_{p_1,p_2})$ in terms of powers of $
\cos \theta_{p_1,p_2}$, using the expressions for the Chebyshev
polynomials in Table 22.3 of Abramowitz \& Stegun (1972). This
enables us to express ${\cal C}(4 m-2,1;s)$ in terms of ${\cal C}(4
n,1;s)$ with $4 n<4 m-2$, and $n$ being a positive integer. In this
way, we can inductively arrive at the results below.

\subsection{Order 0}
The only sum of this type is ${\cal C}(0,1;s)$, given by equation
(\ref{mz2}). It is given by a modified form of (\ref{gr1}), since
there are contributions from both axes $p_1=0$, $p_2=0$ rather than
just $p_2=0$:

\begin{eqnarray}
{\cal C}(0,1;s) =
2\zeta(2 s)+\frac{2\sqrt{\pi}\Gamma(s-1/2)}{\Gamma(s)}\zeta(2 s-1)
\hspace{25ex} \nonumber \\
+\frac{8\pi^s }{\Gamma(s)}\sum_{p_1=1}^\infty \sum_{p_2=1}^\infty
\left( \frac{p_2}{p_1}\right)^{s-1/2} K_{s-1/2}(2\pi p_1 p_2).
\label{sn3}
\end{eqnarray}

\subsection{Order 2}
The single sum of order 2 is
\begin{eqnarray}
\sum_{(p_1,p_2)} {'} \frac{p_1^2}{(p_1^2+p_2^2)^{s+1}} =
{\cal C}(2,1;s)=\frac{1}{2}{\cal C}(0,1;s)\nonumber \\
=\frac{2 \sqrt{\pi} \Gamma(s+1/2)\zeta(2
s-1)}{\Gamma(s+1)}+\frac{8\pi^s}{\Gamma(s+1)} \sum_{p_1=1}^\infty
\sum_{p_2=1}^\infty & & \hspace{-4mm}\left(
\frac{p_2}{p_1}\right)^{s-1/2}\hspace{-4mm}
p_1 p_2 \pi K_{s+1/2}(2\pi p_1 p_2). \nonumber \\
 & &
\label{sn4}
\end{eqnarray}

\subsection{Order 4}
We generate this system from ${\cal C}(4,1;s)$, obtained from
(\ref{gr1}). In terms of this,
\begin{equation}
{\cal C}(1,4;s)=8 {\cal C}(4,1;s)-3{\cal C}(0,1;s), \label{sn6}
\end{equation}
and
\begin{equation}
{\cal C}(2,2;s)=4 {\cal C}(4,1;s)-{\cal C}(0,1;s),~~{\cal
S}(2,2;s)=-4 {\cal C}(4,1;s)+2{\cal C}(0,1;s). \label{sn7}
\end{equation}

\subsection{Order 6}
The result of expanding (\ref{sn1}) for $m=2$ is
\begin{equation}
{\cal C}(6,1;s)=\frac{3}{2}{\cal C}(4,1;s)-\frac{1}{4}{\cal
C}(0,1;s), \label{sn8}
\end{equation}
while from (\ref{sn2})
\begin{equation}
{\cal C}(2,3;s)={\cal S}(2,3;s)=\frac{1}{2}{\cal C}(0,1;s).
\label{sn9}
\end{equation}
Also,
\begin{equation}
\sum_{(p_1,p_2)} {'}\frac{p_1^4
p_2^2}{(p_1^2+p_2^2)^{s+3}}=\sum_{(p_1,p_2)}{'}\frac{p_1^2
p_2^4}{(p_1^2+p_2^2)^{s+3}}= \frac{1}{2}\sum_{(p_1,p_2)}
{'}\frac{p_1^2 p_2^2}{(p_1^2+p_2^2)^{s+2}}, \label{sn10}
\end{equation}
so
\begin{equation}
\sum_{(p_1,p_2)}{'}\frac{\cos^4(\theta_{p_1,p_2})
\sin^2(\theta_{p_1,p_2})}{(p_1^2+p_2^2)^{s}}=\frac{1}{8} {\cal
S}(2,2;s). \label{sn11}
\end{equation}
Note that all sums mentioned so far can be generated from just two,
say ${\cal C}(0,1;s)$ and ${\cal C}(1,4;s)$. ${\cal C}(6,1;s)$ is
given by (\ref{gr1}).

\subsection{Order 8}
The new sum we use here is ${\cal C}(8,1;s)$, obtained from
(\ref{gr1}). In terms of this,
\begin{equation}
{\cal C}(1,8;s)=128 {\cal C}(8,1;s)-224 {\cal C}(4,1;s)+49 {\cal
C}(0,1;s), \label{sn14}
\end{equation}
\begin{equation}
{\cal C}(2,4;s)=64 {\cal C}(8,1;s)-112 {\cal C}(4,1;s)+25 {\cal
C}(0,1;s), \label{sn15}
\end{equation}
and
\begin{equation}
{\cal S}(2,4;s)=-64 {\cal C}(8,1;s)+112{\cal C}(4,1;s)-24 {\cal
C}(0,1;s). \label{sn16}
\end{equation}

Two other sums in this system are
\begin{equation}
\sum_{(p_1,p_2)} {'}\frac{p_1^6 p_2^2}{(p_1^2+p_2^2)^{s+4}}= -{\cal
C}(8,1;s)+\frac{3}{2}{\cal C}(4,1;s)-\frac{1}{4} {\cal C}(0,1;s),
\label{sn17}
\end{equation}
and
\begin{equation}
\sum_{(p_1,p_2)} {'}\frac{p_1^4 p_2^4}{(p_1^2+p_2^2)^{s+4}}= {\cal
C}(8,1;s)-2{\cal C}(4,1;s)+\frac{1}{2} {\cal C}(0,1;s). \label{sn18}
\end{equation}

Using  the result (\ref{sn14}), we have calculated curves showing
the modulus of ${\cal C}(1,8;s)$ as a function of $s=1/2+\ri t$ on
the critical line. These are given in Figs. \ref{supp1}-\ref{supp2},
and the distributions of zeros they show are given in Table~\ref{table1}.

\begin{figure}
\begin{center}
\includegraphics[width=6cm]{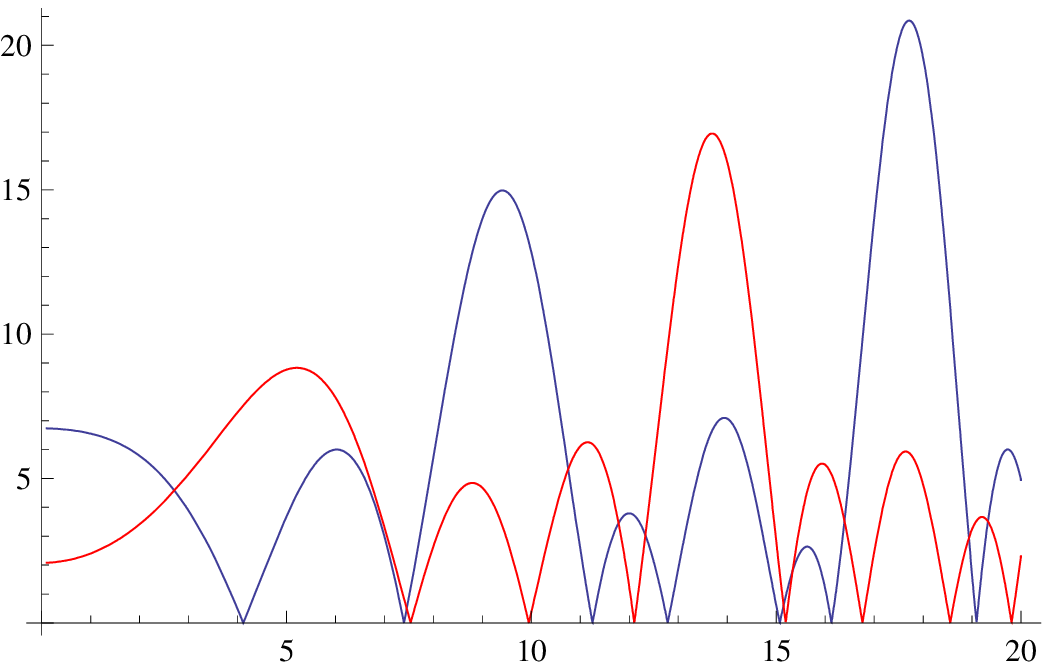}~~
\includegraphics[width=6cm]{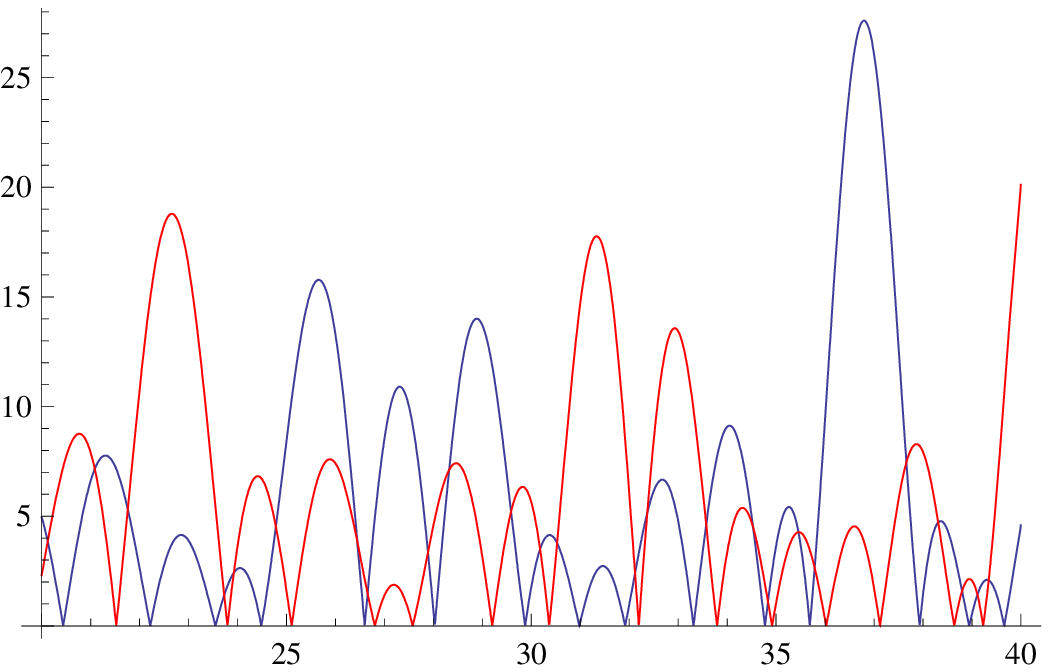}
\caption{The modulus of  ${\cal C}(1,4;s)$ (red) and ${\cal
C}(1,8;s)$ (blue)  as a function of $s=1/2+\ri t$, for $t \in [0, 20]$ (left) and $t \in [20, 40]$ (right).}
 \label{supp1}
\end{center}
\end{figure}

\begin{figure}
\begin{center}
\includegraphics[width=6cm]{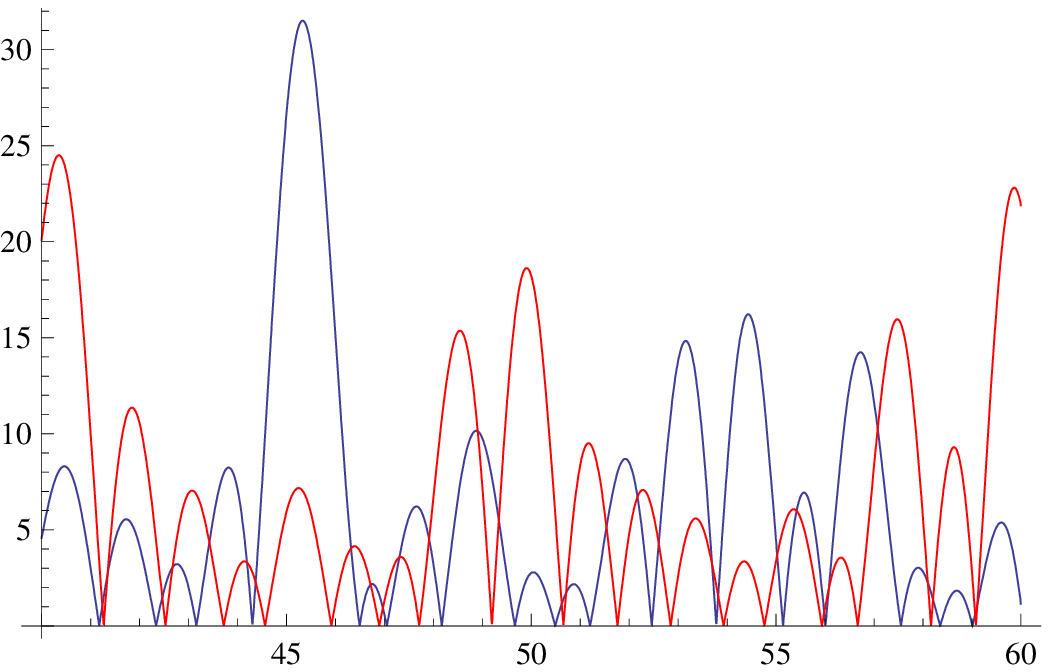}~~
\includegraphics[width=6cm]{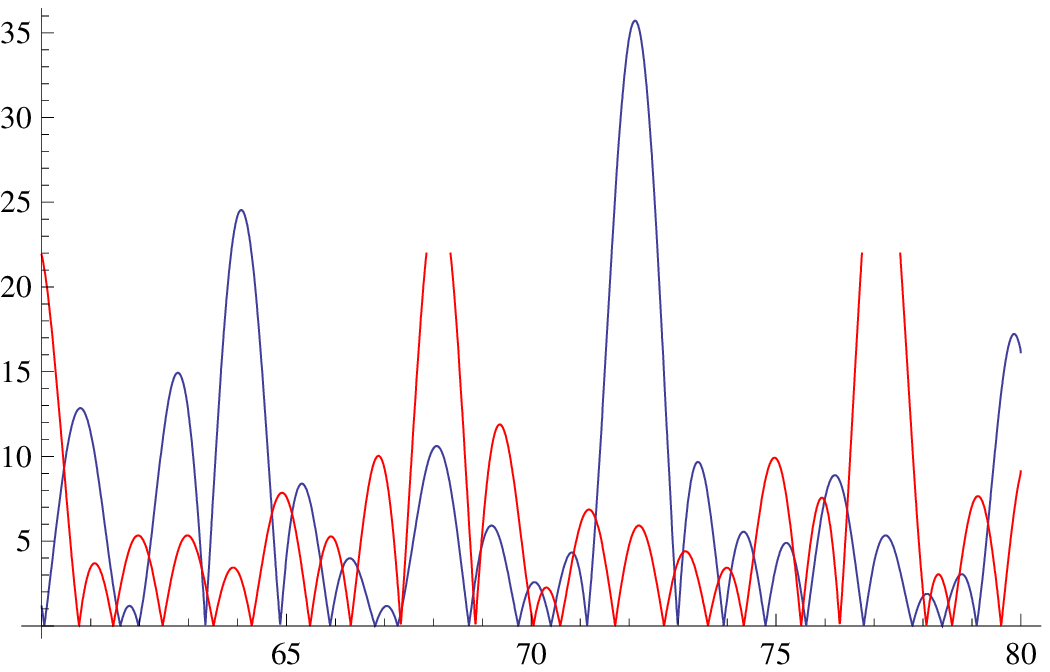}
\caption{The modulus of  ${\cal C}(1,4;s)$ (red) and ${\cal
C}(1,8;s)$ (blue)  as a function of $s=1/2+\ri t$, for $t \in [40, 60]$ (left) and $t \in [60, 80]$ (right).}
 \label{supp2}
\end{center}
\end{figure}

\subsection{Recurrence Relations}
We can generalize the above procedure by establishing recurrence
relations for the trigonometric sums. We consider
\begin{equation}
\sum_{(p_1,p_2)} {'} \frac{p_1^{2 n}}{(p_1^2+p_2^2)^{s+n}}= {\cal
C}(2 n,1;s)=\sum_{(p_1,p_2) \neq (0,0)}
\frac{1}{(p_1^2+p_2^2)^s}\left[1-\frac{p_2^2}{p_1^2+p_2^2}
\right]^n. \label{sn19}
\end{equation}
Expanding using the Binomial Theorem, we obtain
\begin{equation}
{\cal C}(2 n,1;s)=\sum_{l=0}^n {}^nC_l\, (-1)^l \, {\cal C}(2 l,1;s).
\label{sn20}
\end{equation}
If $n$ is even, (\ref{sn20}) gives an identity:
\begin{equation}
\sum_{l=0}^{2 n-1} {}^{2 n}C_l\, (-1)^l \, {\cal C}(2 l,1;s)=0,
\label{sn21}
\end{equation}
while for $n$ odd we obtain an expression for ${\cal C}(4 n-2,1;s)$
in terms of lower order sums:
\begin{equation}
{\cal C}(4 n-2,1;s)=\frac{1}{2}\sum_{l=0}^{2 n-2} {}^{2 n-1}C_l\,
(-1)^l \, {\cal C}(2 l,1;s). \label{sn22}
\end{equation}
It may be checked that the two relations (\ref{sn21}) and
(\ref{sn22}) give equivalent results in the cases given above.

\begin{figure}
\begin{center}
\includegraphics[width=12cm]{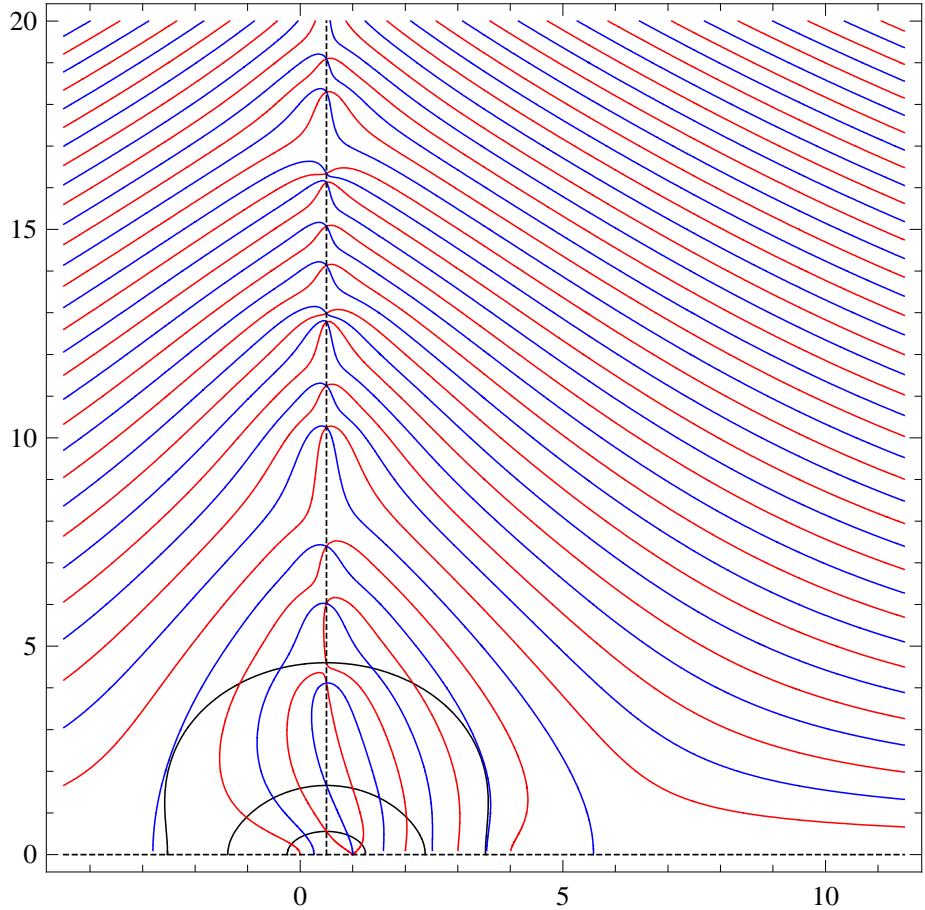}
\end{center}
\caption{Null contours of the real part (red) and imaginary part
(blue) of $\Delta_3(2,4;\sigma+\ri t)$, with $\sigma \in [-4.5, 11.5]$, and
$t \in [0.1, 20]$.} \label{figord8}
\end{figure}

In Fig. \ref{figord8} we give null contours of the real part and
imaginary part of $\Delta_3(2,4;s)$, for comparison with those of
$\Delta_3(2,2;s)$ in Fig. \ref{fig2}. It will be noted that there
are now three lines giving the contours on which $\Imag
\Delta_3(2,4;s)=0$, and these are not circles (as was the single
exemplar in Fig. \ref{fig2}). There is a corresponding increase in
the number of null contours starting and finishing on the real axis,
and in the value of $\sigma$ at which the last one reaches the real
axis. There are two examples in Fig.~\ref{figord8} of zeros on the critical
line of the real part, but not the imaginary part ($\tan{(\phi_{2m, c}(1/2+ \ri t)}= \infty$),
and one of a zero of the imaginary part but not the real part
($\cot{(\phi_{2m, c}(1/2+ \ri t)}= \infty$).
The null contours starting at $\sigma=\infty$ and ending at
$\sigma=-\infty$ settle down to an asymptotic behaviour for somewhat
larger values of $t$ than in Fig. \ref{fig2}, but are otherwise
similar to the lower order null contours.

\subsection{Order 10}
The new sum we use from (\ref{gr1}) is ${\cal C}(10,1;s)$, which,
from (\ref{sn22}), is
\begin{equation}
{\cal C}(10,1;s)=\frac{1}{2}{\cal C}(0,1;s)-\frac{5}{2}{\cal
C}(4,1;s)+\frac{5}{2}{\cal C}(8,1;s). \label{sn24}
\end{equation}
Other sums in the system are:
\begin{equation}
\sum_{(p_1,p_2)} {'} \frac{p_1^{8}p_2^2}{(p_1^2+p_2^2)^{s+5}}={\cal
C}(8,1;s)-{\cal C}(10,1;s), \label{sn25}
\end{equation}
and
\begin{equation}
\sum_{(p_1,p_2)} {'}
\frac{p_1^{6}p_2^4}{(p_1^2+p_2^2)^{s+5}}=\frac{1}{2}\sum_{(p_1,p_2)
} {'} \frac{p_1^4 p_2^4}{(p_1^2+p_2^2)^{s+4}}. \label{sn25x}
\end{equation}

\appendix{Derivation of the Functional Equation for ${\cal C}(1,4 m;s)$}

We start with the equation (30) from McPhedran {\em et al} (2004).
This equation uses the Poisson summation formula to derive a
connection between a sum over a direct lattice of points
$\mathbf{R}_p \equiv \mathbf{R}_{p_1,p_2} = d (p_1, p­_2)$ with polar
coordinates $(R_p,\phi_p)$, and a corresponding sum over the
reciprocal lattice, where the lattice points are labelled
$\mathbf{K}_h \equiv \mathbf{K}_{h_1,h_2}$. Considering the case of a square lattice with period
$d$, the reciprocal lattice points are $\mathbf{K}_h =(2
\pi/d)(h_1,h_2)$. The sum in the direct lattice incorporates a phase
term of the Bloch type, with wave vector $\mathbf{k}_0$, and the result
of the Poisson formula is
\begin{eqnarray}
2^{2 s-1} \Gamma(\frac{l}{2}+s)\sum_{(p_1,p_2)} {'}\frac{\re^{\ri
\mathbf{k}_0 \cdot \mathbf{R}_p} \re^{\ri l\phi_p}}{R_p^{2 s}}
\hspace{35ex}
\nonumber \\
= \frac{2 \pi \ri^l}{d^2} \Gamma(\frac{l}{2}+1-s) \left(
\frac{\re^{\ri l\theta_0}}{k_0^{2-2s}}\sum_{(h_1,h_2)}
{'}\frac{\re^{\ri l\theta_h}}{Q_h^{2-2s}}\right) . \label{poisson1}
\end{eqnarray}
Here the sum over the reciprocal lattice in fact runs over a set of
displaced vectors:
\begin{equation}
{\bf Q}_h={\bf k}_0+{\bf K}_h=(\frac{2 \pi h_1}{d}+k_{0x},\frac{2
\pi h_2}{d}+k_{0y})=(Q_h,\theta_h), \label{poisson2}
\end{equation}
where in (\ref{poisson2}) the second and third expressions are in
rectangular and polar coordinates. Note that in (\ref{poisson1}) and
(\ref{poisson2}), $p_1, p_2$ and $h_1, h_2$ run over all integer
values, and $\theta_0$ gives the direction of ${\bf k}_0$.

We express the relation (\ref{poisson1}) in non-dimensionalized
form, taking out a factor $d^{2 s}$ on the left-hand side, and a
factor $(2 \pi/d)^{2-2s}$ on the right-hand side. We put ${\bf
k}_0=(2 \pi/d) \mbox{\boldmath $\kappa$}_0$, and obtain
\begin{eqnarray}
\Gamma(\frac{l}{2}+s) \sum_{(p_1,p_2)} {'}\frac{\re^{2 \pi \ri
(\kappa_{0x}p_1+\kappa_{0y}p_2)} \re^{\ri l\phi_p}}{(p_1^2+p_2^2)^{
s}}=\ri^l \pi^{2 s-1} \Gamma(\frac{l}{2}+1-s) \hspace{8ex} \nonumber \\
\times \left( \frac{\re^{\ri
l\theta_0}}{\kappa_0^{2-2s}}+\sum_{(h_1,h_2)} {'}\frac{\re^{\ri
l\theta_h}}{((\kappa_{0x}+h_1)^2+(\kappa_{0y}+h_2)^2)^{1-s}}\right)
.\label{poisson3}
\end{eqnarray}
If we now let $\kappa_0\rightarrow 0$, set $l=4 m$ and take $\Real
(s)>1$, we obtain the desired result:
\begin{equation}
\frac{\Gamma (2 m+s)}{\pi^s} {\cal C}(1, 4 m;s)=\frac{\Gamma (2
m+1-s)}{\pi^{1-s}} {\cal C}(1, 4 m;1-s) .\label{poisson4}
\end{equation}

\label{lastpage}
\end{document}